\documentclass[aos]{imsart}
\RequirePackage{amsthm,amsfonts,amssymb}
\RequirePackage[numbers,sort]{natbib}
\RequirePackage[colorlinks,citecolor=blue,urlcolor=blue]{hyperref}
\RequirePackage{graphicx}

\usepackage[reqno]{amsmath}
\usepackage{caption}
\usepackage{subcaption}
\usepackage{enumerate}% allows enumerate with latin characters
\usepackage{dsfont}% allows doublestroke symbols with \mathds{}
\usepackage[usenames,dvipsnames]{color}
\usepackage{mathrsfs}
\usepackage{placeins}
\usepackage{graphicx}
\usepackage{tikz}
\usetikzlibrary{arrows}
\usepackage{url}
\usepackage{comment}
\usepackage{multicol}
\usepackage{cleveref}

\usepackage[font=small,labelfont=small]{caption}

\usepackage{float}
\restylefloat{table}

\usepackage{natbib}

\renewcommand{\thetable}{\arabic{section}.\arabic{table}}
\numberwithin{equation}{section}

\usepackage{aligned-overset}
\usepackage{bbm}
\usepackage{longtable} % For long tables
\usepackage{booktabs}
\crefname{assumptionletter}{Model}{Model}
\Crefname{assumptionletter}{Model}{Model}
\usepackage{bibunits}

\defaultbibliographystyle{imsart-nameyear} 
\usepackage{etoolbox} %% Underfull HBox in .bbl
\apptocmd{\sloppy}{\hbadness 10000\relax}{}{}

\crefname{equation}{}{} %No equation when referencing equations
\Crefname{equation}{}{} %No equation when referencing equations
\theoremstyle{plain} %default (text in italic)
\newtheorem{theorem}{Theorem}[section]
\newtheorem{lemma}[theorem]{Lemma}
\newtheorem{proposition}[theorem]{Proposition}
\newtheorem{corollary}[theorem]{Corollary}
\newtheorem{assumptionletter}{{Model}}

\theoremstyle{remark}
\newtheorem{definition}[theorem]{Definition}

\newtheorem{example}[theorem]{Example}

\newtheorem{remark}[theorem]{Remark}

\def\tablename{\small \textsc{Table}}

\newcommand{\bthe}{\begin{theorem}}
\newcommand{\ethe}{\end{theorem}}

\newcommand{\ben}{\begin{enumerate}}
\newcommand{\een}{\end{enumerate}}

\newcommand{\bit}{\begin{itemize}}
\newcommand{\eit}{\end{itemize}}

\newcommand{\beq}{\begin{equation}}
\newcommand{\eeq}{\end{equation}}

\newcommand{\ble}{\begin{lemma}}
\newcommand{\ele}{\end{lemma}}

\newcommand{\bde}{\begin{definition}\rm}
\newcommand{\ede}{\halmos\end{definition}}

\newcommand{\bco}{\begin{corollary}}
\newcommand{\eco}{\end{corollary}}

\newcommand{\bpr}{\begin{proposition}}
\newcommand{\epr}{\end{proposition}}

\newcommand{\brem}{\begin{remark}\rm}
\newcommand{\erem}{\end{remark}}

\newcommand{\bproof}{\begin{proof}}
\newcommand{\eproof}{\end{proof}}

\newcommand{\bexam}{\begin{example}\rm}
\newcommand{\eexam}{\end{example}}

\newcommand{\bfi}{\begin{fig}}
\newcommand{\efi}{\end{fig}}

\newcommand{\btab}{\begin{tab}}
\newcommand{\etab}{\end{tab}}

\newcommand{\beao}{\begin{eqnarray*}}
\newcommand{\eeao}{\end{eqnarray*}\noindent}

\newcommand{\balo}{\begin{align*}}
\newcommand{\ealo}{\end{align*}}

\newcommand{\balm}{\begin{align}}
\newcommand{\ealm}{\end{align}\noindent}

\newcommand{\beam}{\begin{eqnarray}}
\newcommand{\eeam}{\end{eqnarray}\noindent}

\newcommand{\barr}{\begin{array}}
\newcommand{\earr}{\end{array}}

\newcommand{\E}{\mathbb{E}}
\newcommand{\V}{\mathbb{VAR}}

\newcommand{\N}{\mathbb{N}}
\renewcommand\P{\mathbb{P}}
\newcommand{\R}{\mathbb{R}}

\newcommand{\vertiii}[1]{{\left\vert\kern-0.25ex\left\vert\kern-0.25ex\left\vert #1 
    \right\vert\kern-0.25ex\right\vert\kern-0.25ex\right\vert}}

\def\bV{\boldsymbol{V}}
\def\bI{\boldsymbol{I}}
\def\bSigma{\boldsymbol{\Sigma}}
\def\bSigman{{\boldsymbol{\Sigma}}^{(n)}}
\def\bGamma{\boldsymbol{\Gamma}}
\def\bGamman{\boldsymbol{\Gamma}_n}
\def\bUpsilon{\boldsymbol{\Upsilon}}
\def\bepsilon{\boldsymbol \varepsilon}

\def\bA{{\boldsymbol A}}
\def\bAn{{\boldsymbol A}^{(n)}}
\def\bB{\boldsymbol B}
\def\bBn{{\boldsymbol B}^{(n)}}

\def\bS{\boldsymbol S}
\def\bu{\boldsymbol u}

\def\bV{\boldsymbol V}
\def\bv{\boldsymbol v}
\def\bVn{{\boldsymbol V}^{(n)}}
\def\bX{\boldsymbol X}
\def\bXn{{\boldsymbol X}^{(n)}}

\def\bb{\boldsymbol b}

\def\bTn{{\boldsymbol T}^{(n)}}
\def\bW{ {\boldsymbol W}^{(n)} }
\def\bw{\boldsymbol w}

\def\bCn{{{\boldsymbol P}^{(n)}}}

\def\bM{\boldsymbol M}
\def\bm{\boldsymbol m}
\def\bN{\boldsymbol N}
\def\bH{\boldsymbol H}
\def\bP{\boldsymbol P}

\def\wbVn{{\mathcal{ V}}^{(n)}}

\def\sqrtsbSigman{ \sbSigman^{1/2} }
\def\wbTheta{  \sqrtsbSigman  \wbVn \bTn}

\def\bx{\boldsymbol x}

\def\ba{\boldsymbol a}
\def\be{\boldsymbol e}

\def\Un{U^{(n)}}
\def\Vn{V}

\def\obWn{{\overline{\boldsymbol{W}}^{(n)}}}
\def\tbWn{{\widetilde{\boldsymbol{W}}^{(n)}}}
\def\tbWnT{{\widetilde{\boldsymbol{W}}^{(n)\,\top}}}
\def\bTheta{\boldsymbol \Theta}
\def\bThetan{{\boldsymbol \Theta}^{(n)}}
\def\btheta{\boldsymbol \theta}

\def\bXin{{\boldsymbol{D}^{(n)}}}
\def\obXi{\overline{\boldsymbol{D}}_{n}}

\def\bhSigma{\widehat{\boldsymbol \Sigma }}
\def\bhSigman{{\widehat{\boldsymbol \Sigma }}{}^{(n)}}
\def\bhSigmanprime{\bhSigma{}^{(n)\prime}}
\def\sbhSigman{ \bUpsilon^{(n)}} %{\widehat{\bGamma}_n} \widehat{\bGamma}   \bhSigma^*_n
\def\sbSigman{{{\boldsymbol{\Gamma}}^{(n)}}}   % \bGamma_n \bSigma^*_n

\def\well{\widehat{\lambda}}

\def\wellprime{\well'}
\def\hslambda{\widehat{\xi}}
\newcommand{\slambda}[1]{\xi_{n,#1}}

\def\slambdaohnen{\xi}

\def\bhSigmanprimesub{\bhSigma{}^{(n)\prime}}
\newcommand{\slambdasub}[1]{\xi_{n,#1}}

\def\obWnsub{{\overline{\boldsymbol{W}}^{(n)}}}
\def\tbWnsub{{\widetilde{\boldsymbol{W}}^{(n)}}}
\def\tbWnTsub{{\widetilde{\boldsymbol{W}}^{(n)\,\top}}}

\def\bAnsub{{\boldsymbol A}^{(n)}}
\def\bBnsub{{\boldsymbol B}^{(n)}}

\def\bTnsub{{\boldsymbol T}^{(n)}}
\def\bWsub{ {\boldsymbol W}^{(n)} }
\def\bCnsub{{{\boldsymbol P}^{(n)}}}

\def\wbVnsub{{\mathcal{V}}^{(n)}}
\def\dnsub{d_{n}}
\def\bXinsub{{\boldsymbol{D}^{(n)}}}

\DeclareMathOperator*{\argmin}{arg\,min}

\newcommand{\vague}{\stackrel{\lower0.2ex\hbox{$\scriptscriptstyle
                    \it{v} $}}{\rightarrow}}
\newcommand{\weak}{\stackrel{\lower0.2ex\hbox{$\scriptscriptstyle
                    \it{w} $}}{\rightarrow}}
\newcommand{\what}{\stackrel{\lower0.2ex\hbox{$\scriptscriptstyle
                    \it{\hat{w}} $}}{\rightarrow}}
\newcommand{\eqdis}{\stackrel{\lower0.2ex\hbox{$\scriptscriptstyle
                    \mathrm{d}$}}{=}}
\newcommand{\distr}{\stackrel{\lower0.2ex\hbox{$\scriptscriptstyle
                    \it{d} $}}{\rightarrow}}
\DeclareMathOperator{\AIC}{AIC}
\DeclareMathOperator{\BIC}{BIC}
\newcommand{\QAIC}{\AIC^\circ}
\newcommand{\QBIC}{\BIC^\circ}
\newcommand{\QAICstar}{\AIC^*}
\newcommand{\QBICstar}{\BIC^*}

\newcommand{\vX}{\lVert \bX \rVert}
\newcommand{\vXi}{\lVert \bX_i \rVert} 
\newcommand{\ninf}{n \to \infty}

\newcommand{\limn}{\lim_{n \rightarrow \infty}}

\newcommand{\asconv}{\overset{\P \text{-a.s.}}{\longrightarrow}}
\newcommand{\Pconv}{\overset{\mathbb{P}}{\longrightarrow}}
\newcommand{\Dconv}{\overset{\mathcal{D}}{\longrightarrow}}
\newcommand{\di}{\, \mathrm{d}}

\newcommand{\ps}{p^*}
\newcommand{\pn}{p_n}

\renewcommand{\hat}{\widehat}

\DeclareMathOperator{\rank}{rank}
\DeclareMathOperator{\Cov}{Cov}
\DeclareMathOperator{\ve}{vec}
\DeclareMathOperator{\diag}{diag}

\allowdisplaybreaks %% display breaks for align
\definecolor{darkgreen}{RGB}{0,139,0}

\begin{document}
\begin{bibunit}

\begin{frontmatter}
%\title{Multivariate regular variation and the Gaussian copula}
%\title{On Gaussian copula, asymptotic independence, and heavy-tailed marginals}
%\title{Information criteria for the dimension of \vspace*{0.2cm} \\  PCA of extremes in fixed- and high-dimensional data}
\title{Estimation of the number of principal components \vspace*{0.2cm}\\ in high-dimensional multivariate extremes}
\runtitle{PCA for multivariate extremes}

\begin{aug}
  \author{\fnms{Lucas} \snm{Butsch}\ead[label=e1]{lucas.butsch@kit.edu} }%\orcid{0000-0002-xxxx-xxxx}}
    \and
   \author{\fnms{Vicky} \snm{Fasen-Hartmann}\ead[label=e2]{vicky.fasen@kit.edu}\orcid{0000-0002-5758-1999}}
%\thanksref{t1}\thankstext{t1}
 \address{Institute of Stochastics, Karlsruhe Institute of Technology \\[2mm] \printead[presep={ }]{e1,e2}}

%  \thanksref{T1}
%  \thankstext{T1}{}

  \runauthor{L. Butsch and  V. Fasen-Hartmann}
\end{aug}

\begin{abstract}

For multivariate regularly random vectors of dimension $d$, the dependence structure of the extremes is modeled by the so-called angular measure. When the dimension $d$ is high, estimating the angular measure is challenging because of its complexity.  In this paper, we use Principal Component Analysis (PCA) as a method for dimension reduction and estimate the 
number of significant principal components of the empirical covariance matrix of the angular measure under the assumption of a spiked covariance structure. Therefore, we develop Akaike Information Criteria ($\AIC$) and  Bayesian Information Criteria ($\BIC$) to estimate the location of the spiked eigenvalue of the covariance matrix, reflecting the number of significant components, and explore these information criteria on consistency. On the one hand, we investigate the case where the dimension $d$ is fixed, and on the other hand, where the dimension $d$ converges to $\infty$ under different high-dimensional scenarios. 
When the dimension $d$ is fixed, we establish that the $\AIC$ is not consistent, whereas the $\BIC$ is weakly consistent. In the high-dimensional setting, with techniques from random matrix theory, we derive sufficient conditions for the $\AIC$ and the $\BIC$ to be consistent. Finally, the performance of the different $\AIC$ and $\BIC$ versions is compared in a simulation study and applied to high-dimensional precipitation data.

\end{abstract}



%\begin{keyword}[class=JEL]
%\kwd[Primary ]{C14}
%\kwd{C61}
%\kwd{C63}
%\kwd{D81}
%\kwd{G11} 
%\kwd{G22}
%\iffalse
%\kwd[Primary ]{60F10}
%\kwd{60G50}
%\kwd{60G70}
%\kwd[; secondary ]{60B10}
%\kwd{62G32}
%\fi
%\end{keyword}

\end{frontmatter}
%\maketitlehttps://arxiv.org/abs/math.PR/0000000

%==================================================================================================
\section{Introduction}
%==================================================================================================

In multivariate extreme value theory, extremes occur per se rarely so that the dimensionality of the data in fields such as finance, insurance, meteorology, hydrology, and, more broadly, environmental risk assessment often approaches or exceeds the number of extreme observations which is a big challenge in the statistical analysis of complex and high-dimensional data. 
Therefore, a standard approach from multivariate statistics is to apply a dimension reduction method to reduce the model complexity and circumvent the curse of dimensionality. %This drastically reduces the computational time and the quality of the estimation of the angular measure. 
A very nice overview of different methods for constructing sparsity in high-dimensional multivariate extremes is given in  \citet{Engelke:Ivanovs}, including Principle Component Analysis (PCA),  spherical $k$-means, graphical models, or sparse regular variation, to mention a few.

In multivariate statistics,  PCA is a widely used method for dimension reduction, data visualization, clustering and feature extraction
(\cite{muirhead:1982,anderson:2003}).  
In recent years, the literature on implementing PCA for high-dimensional and complex data of multivariate extremes to construct some sparsity in the data has grown rapidly 
(\cite{chautru,DS:21,Sabourin_et_al_2024,
CT:19,MR4582715, drees2025asymptoticbehaviorprincipalcomponent,wan2024characterizing}). A classical concept of multivariate extreme value theory is multivariate regular variation (\cite{resnick:1987,resnick:2007,Falk:Buch}). 
 A $d$-dimensional random vector $\bX$ is \textit{multivariate regularly varying }of index $\alpha>0$ if there exists a random vector $\bTheta$ on the unit sphere  such that
\begin{equation*} 
\P \left( \frac{\vX}{t} > r, \frac{\bX}{\vX} \in \cdot \Big| \, \vX > t \right) \Dconv r^{-\alpha} \P( \bTheta \in \cdot),\quad t \rightarrow \infty,
\end{equation*}
for all $r > 0$. 
The dependence structure of the extremes of $\bX$ is modeled in the 
spectral vector $\bTheta$ whose distribution is also called \textit{angular measure} and its covariance matrix is denoted as $\bSigma=\Cov(\bTheta)$. \citet{drees2025asymptoticbehaviorprincipalcomponent}  and
\citet{DS:21} set the mathematical framework for PCA for the empirical covariance estimator of $\bSigma=\Cov(\bTheta)$ by analyzing the squared reconstruction error, the excess risk and their asymptotic behavior.
However, until now, the research on the number of significant principal components, the so-called  \textit{dimensionality}, in PCA for multivariate extremes, is limited. In \citet{DS:21}, the dimension was estimated through the examination of empirical risk plots.  An alternative approach is to analyze the scree plot, which is the plot of the empirical eigenvalues, in search of an "elbow" as a cutoff point, indicating a minimal variation in the empirical eigenvalues after this point. But a big challenge in extreme value theory is the choice of the threshold $t$, which defines the extreme observations as the data whose norm is above $t$. Changing this threshold also changes the number of extreme observations and the estimates of the empirical eigenvalues. Thus, a change of the threshold results in a different scree plot and possibly also in a different elbow.  Given the ambiguity and uncertainty of both methods, as well as the need for case-by-case evaluation, a mathematically based approach is necessary.
One of the few works that developed a statistical method for estimating the dimensionality is given in \citet{drees2025asymptoticbehaviorprincipalcomponent}, whose method is based on asymptotic results for the reconstruction error of the projections and is a kind of testing problem with the disadvantage that it depends on different tuning parameters.

In this paper, we propose information criteria to estimate the number of significant principal components in multivariate extremes modeled through the covariance matrix $\bSigma$ of $\bTheta$. Therefore, we combine an approach of 
\citet{Bai:Choi:Fujikoshi:2018} and \citet{jiang:2023} from
high-dimensional statistics with methods from extreme value theory (\cite{resnick:1987,resnick:2007,HF:2006}) and random matrix theory (\cite{Bai:Silverstein:2010}). We assume a \textit{spiked covariance model} for the covariance matrix $\bSigma$ which goes back to \citet{johnstone:2001} and is widely used in high-dimensional statistics (\cite{Bai:Choi:Fujikoshi:2018,Bai:Yao:2012,fujikoshi:sakurai:2016,jiang:2023,johnstone:2018}) with applications in various fields, e.g., speech recognition, wireless communication and statistical learning as mentioned in 
\citet{Paul:2007}. 

\smallskip

\noindent\textsc{Spiked Covariance Model}: \hspace{0.2cm}
\textit{The eigenvalues $\lambda_1, \ldots, \lambda_d$ of $\bSigma$ satisfy}
\begin{align}
    \lambda_1 \ge \lambda_2 \ge \cdots \ge \lambda_{\ps} > \lambda_{{\ps}+1} = \cdots = \lambda_{d-1} \eqqcolon \lambda > 0 . \label{sec:model}
\end{align}

\smallskip

The smallest eigenvalue $\lambda_d$ is not considered to avoid numerical instability as it can be equal to $0$, if, for example, $\bTheta$ is concentrated on the subspace of the unit sphere with non-negative values.  
The main objective of this article is to develop an estimator $\widehat{p}_n$ for the unknown dimensionality parameter $\ps$  of significant eigenvalues of $\bSigma$  through information criteria. When $d$ is relatively large compared to $\widehat{p}_n$, the data are projected on the $\widehat{p}_n$-dimensional subspace spanned by the empirical eigenvectors of the largest  $\widehat{p}_n$ empirical eigenvalues $\well_{n,1}, \ldots, \well_{n,\widehat{p}_n}$. This lower-dimensional representation allows for more extensive and in-depth analyses of the dependence structure in the extremes.

The estimation of the dimensionality $\ps$ in PCA is explored in this paper using two information criteria: the Akaike Information Criterion ($\AIC$) and the Bayesian Information Criterion ($\BIC$). Similar information criteria were investigated in \citet{fujikoshi:sakurai:2016} for Gaussian random vectors in a large-sample asymptotic framework and in
  \citet{Bai:Choi:Fujikoshi:2018} for general data in the high-dimensional case. Both criteria are motivated by a Gaussian likelihood function, although our underlying model for $\bTheta$ is not Gaussian.  Since in a Gaussian spiked covariance model with $k_n$ observations, the log likelihood function can be written as a functional of the empirical eigenvalues in the form
  \begin{equation*}
      k_n\sum_{i=1}^{p^*} \log  ( \well_{n,i}  )  + k_n(d - p^*)  \log  \Big(  \sum_{j = p^*}^{d} \frac{\well_{n,j}}{d_n -p^*} \Big) + k_n  \log \Big( \frac{k_n -1}{k_n} \Big)^d + k_n d( \log(2 \pi) + 1),
  \end{equation*} 
   both the AIC and BIC are defined as functionals of the empirical eigenvalues $\well_{n,1}, \ldots, \well_{n,d}$. 

  The main goal of this paper is to derive necessary and sufficient conditions for our AIC and BIC to be consistent. Therefore, we require methods from random matrix theory to derive the asymptotic properties of the empirical eigenvalues, which are the basic components of the AIC and BIC. For this purpose,
  we differentiate between two cases when $n$ observations are available and of these $k_n$ are extreme. The first is the classic large sample size and fixed-dimension case, where $\ninf$ and the dimension $d$ is fixed. As is typical for such information criteria, we find that the BIC is consistent, whereas the AIC is, in general, not consistent. In the second case, we assume that $d = d_n$ also depends on $n$ and $d_n / k_n \rightarrow c >0$ as $\ninf$. In this case, the empirical eigenvalues are not consistent estimators for the eigenvalues anymore.  For high-dimensional i.i.d. data with finite fourth moments, it is well-known 
that the \textit{empirical spectral distribution function} converges to the Mar\v{c}enko-Pastur law (\cite{marchenko:pastur:1967}), which describes the bulk distribution of the empirical eigenvalues. The spiked covariance model, introduced by \citet{johnstone:2001}, extends the Mar\v{c}enko-Pastur framework by adding a small number of spiked eigenvalues corresponding to relevant dimensions for the PCA.  In the context of this paper, we derive as well the asymptotic properties of the empirical eigenvalues of $\bSigma$
with the Mar\v{c}enko-Pastur distribution in the limit
and use it for the investigation of the consistency of our information criteria.
To the best of our knowledge, this paper is the first paper to develop consistent information criteria for the dimensionality $\ps$ of the PCA in high-dimensional multivariate extremes. The only other information criteria of \citet{meyer_muscle23} and  \citet{butsch2024fasen} use the concept of sparse regular variation to construct sparsity in the data, in contrast to PCA.

\subsection*{Structure of the paper} This paper is organized as follows. 
In \Cref{sec:asymptotic}, we properly define the empirical eigenvalues $\well_{n,1}, \ldots, \well_{n,d}$ of $\bSigma$, which are the main components in the definition of the information criteria. In addition, we explore the asymptotic properties of the empirical eigenvalues, where in the high-dimensional case we restrict our study to a parametric family of distributions, the so-called \textit{directional model}. The subject of \Cref{sec:fixed_dim} are the $\AIC$  and the $\BIC$  for estimating 
the location $p^*$ of the spiked eigenvalue in the fixed-dimensional case, where \Cref{sec:high_dim} explores the high-dimensional case when $d_n / k_n \rightarrow c >0$ as $\ninf$.
We will examine  the case $0<c<1$ and $c>1$ separately in 
\Cref{sec:high_dim_d_klein}  and
\Cref{sec:high_dim_d_groß}, respectively. In both cases, we derive sufficient criteria for the 
$\AIC$ and the $\BIC$ to be weakly consistent.
In a simulation study in \Cref{sec:simulation}, we compare the different information criteria and apply them to precipitation data in \Cref{sec:data_set}. Finally, we state a conclusion in \Cref{sec:conclusion}. The proofs for the results presented in this paper are provided in the Appendix.

\subsection*{Notation} Throughout the paper, we use the following notation and assume that all random variables are defined on the same probability space $(\Omega,\mathcal{F},\P)$.   First of all, $\Vert \bx \Vert$ is the Euclidean norm for $\bx\in\R^d$ and $\Vert \bA \Vert$ is the spectral norm for matrices $\bA\in\R^{d\times d}$.  
 The matrix $\bI_d\in\R^{d\times d}$ is  the  identity matrix, $\be_i$ is the $i$-th unit vector with  $1$ at the $i$-th entry and $0$ else, $\boldsymbol{0}_d \coloneqq (0, \ldots, 0)^\top \in \R^d$ is the zero vector  and $\boldsymbol{1}_d: = (1, \ldots, 1)^\top \in \R^d$ is the vector  containing only $1$. 
 For a vector $\bx\in \R^d$ we write $\diag(\bx)\in\R^{d\times d}$ for a diagonal matrix with the components of $\bx$ on the diagonal and for $\bA = (\ba^{(1)}, \ldots, \ba^{(d)})   \in \R^{d \times d}$ the operator $\ve (\bA) \in \R^{d^2}$ stacks the columns of $\bA$ in a vector such that $\ve (\bA) = ({\ba^{(1)}}^\top, \ldots, {\ba^{(d)}}^\top)^\top $ and $\lambda_i(\bA)$ denotes the $i$-th largest eigenvalue of $\bA$.  If $\bB  \in \R^{d \times d}$ and $\bA = \bB^2 \in \R^{d \times d}$ then $\bA^{1/2} \coloneqq \bB$ denotes the square root of a matrix. A sequence of matrices $\bA_1, \bA_2, \ldots \in \R^{d \times d}$ with fixed dimension $d$ is denoted by $(\bA_n)_{n \in \N}$ and if the dimensions $d = d_n$ depends on $n$, we write $(\bAn)_{n \in \N}$, where $\bAn \in \R^{d_n \times d_n}$. For a univariate distribution function $F$ the function $F^\leftarrow : (0,1) \rightarrow \R $ with $  p \mapsto \inf \{ x \in \R: F(x) \ge p \}$ is the  generalized inverse of $F$. By $\delta_{\bx}$ we denote the Dirac measure in $\bx \in \R^{d}$.
  Finally,  $\Dconv$ is the notation for convergence in distribution,  $\Pconv$ is the notation for convergence in probability and  $\asconv$ is the notation for almost sure convergence. 

\section{Asymptotic behavior of the empirical eigenvalues of \texorpdfstring{$\bSigma$}{Sigma}} \label{sec:asymptotic}

The information criteria AIC and BIC of this paper are defined by the empirical eigenvalues  $\well_{n,1}, \ldots, \well_{n,d}$ of $\bSigma$.   Therefore, in the first step, in \Cref{sec:fixed_dim_asymptotic}, we define and explore the empirical eigenvalues and their asymptotic properties in the fixed-dimensional case, and then, in \Cref{sec:high_dim_asymptotic}, in the high-dimensional case. With the knowledge of the asymptotic behavior of the empirical eigenvalues,  we will be able to derive the asymptotic behavior of the AIC and the BIC in \Cref{sec:fixed_dim} and \Cref{sec:high_dim}. The proofs of this section are moved to \Cref{sec:asymptotic_proof}.

\subsection{Fixed-dimensional case}  \label{sec:fixed_dim_asymptotic}

In the case where the dimension $d$ is fixed, we consider the following model. 

\begin{assumptionletter}~ \label{asu:fixed_dim}
\begin{enumerate}
    \item[(A1)] Let $\bX, \bX_1, \bX_2, \ldots$ be a sequence of i.i.d. regularly varying random vectors with tail index $\alpha>0$ and spectral vector $\bTheta$.
    \item[(A2)] The eigenvalues $\lambda_1, \ldots, \lambda_d$ of $\bSigma=\Cov(\Theta)$ satisfy
\begin{align*}
    \lambda_1 \ge \lambda_2 \ge \cdots \ge \lambda_{\ps} > \lambda_{{\ps}+1} = \cdots = \lambda_{d-1} \eqqcolon \lambda > 0 . 
\end{align*}
    \item[(A3)] Let $(k_n)_{n \in \N}$ be a sequence in $\N$ with $k_n \rightarrow \infty$ and $k_n / n  \rightarrow 0$ for $ n \rightarrow \infty$.
    \item[(A4)] \label{(A4)} Suppose  $u_n$ is a sequence such that for $\ninf$, $n \P( \Vert \bX \Vert > u_n) / k_n \to 1$ and
    \begin{align*}
       \sup_{x \in [\frac{1}{1 + \tau}, 1 + \tau]} \sqrt{k_n}  \left\Vert   \frac{ n }{k_n}\E \left[ \begin{pmatrix}
            \frac{ \ve (\bX \bX^\top)}{\vX^2}\\
            1
        \end{pmatrix}  \mathbbm{1} \{ x \vX > u_n \} \right]  - x^\alpha \begin{pmatrix}
            \ve (\E[\bTheta \bTheta^\top] )\\
            1
        \end{pmatrix}  \right\Vert \rightarrow 0,  %\label{asu:conv_theta_cov}
   \end{align*}
   as $\ninf$.
\end{enumerate}
\end{assumptionletter}
The last assumption (A4) is a technical assumption that we require for some proofs (cf. \Cref{Remark 2.2}). 
Under these model assumptions, the empirical estimator for $\bTheta$ is defined as
\begin{align*}
   \widehat{\bTheta}_n  \coloneqq \frac{1}{k_n} \sum_{i=1}^n  \frac{\bX_i}{\vXi} \mathbbm{1}\{ \vXi >  \Vert \bX_{(k_n+1,n)} \Vert  \},
\end{align*}
and hence,  the empirical covariance matrix $\bhSigma_n $ of $\bSigma$ is  
\begin{align}
     \bhSigma_n &\coloneqq    \frac{1}{k_n} \sum_{j=1}^n  \left( \frac{\bX_j}{\Vert \bX_j \Vert} -  \widehat{\bTheta}_n  \right) \left( \frac{\bX_j}{\Vert \bX_j \Vert} -  \widehat{\bTheta}_n  \right)^\top \mathbbm{1} \{ \Vert \bX_j \Vert > \Vert \bX_{(k_n+1,n)}\Vert  \} \label{def:sigma_estimator}
\end{align}
with eigenvalues $\well_{n,1}, \ldots, \well_{n,d}$ where $n\in\N$ is the number of observations and $\bX_{(k_n+1,n)}$ denotes the observation with the $(k_n+1)$-th largest norm. 
Both the AIC and the BIC information criteria for the estimation of $\ps$  will be defined by the empirical eigenvalues $\well_{n,1}, \ldots, \well_{n,d}$. Therefore, it is important to know the asymptotic behavior. We start to derive the asymptotic behavior of the empirical covariance matrix $\bhSigma_n $ in the next proposition and use this to derive the asymptotic behavior of the empirical eigenvalues. 

\begin{proposition} \label{th:asymp_theta_cov} 
 Let \Cref{asu:fixed_dim} be given.  
   Then  as $\ninf$,
    \begin{equation*}
       \sqrt{k_n} \bigl( \bhSigma_n - \bSigma \bigr) \Dconv \bS,
    \end{equation*}
  where $\ve (\bS)$ follows a centered normal distribution with covariance matrix 
\begin{equation*}
    \Cov \bigl( \ve( (\bTheta - \E[\bTheta])(\bTheta - \E[\bTheta])^\top  )  \bigr).
\end{equation*}
\end{proposition}

\begin{remark} \label{Remark 2.2}
In the bivariate case and for $h: \R^2\mapsto \R $ defined as $ h(x,y)=xy$, the asymptotic distribution  of 
\begin{align*}
    \frac{1}{k_n} \sum_{i=1}^n  h \left( \frac{\bX_i}{\vXi} \right) \mathbbm{1}\{ \vXi >  \Vert \bX_{(k_n+1,n)} \Vert  \}
\end{align*}
 was derived in \citet[Theorem 1]{Larsson:Resnick:2012}. The techniques of the proof can be generalized and applied to $\ve (\bhSigma_n)$ with the technical assumption (A4), and therefore, the proof of \Cref{th:asymp_theta_cov} is omitted. Note that if $\Vert \btheta \Vert = 1$ for $\btheta \in \R^d$ then  $\Vert \ve(\btheta \btheta^\top) \Vert = 1$ and  higher moments of $\bTheta$ exist, since $\bTheta$ is bounded. A complementary result on the asymptotic behavior of the empirical covariance matrix is also given in the recent publication 
 \citet[Theorem 2.1]{drees2025asymptoticbehaviorprincipalcomponent}.
\end{remark}

Now, we are able to present the asymptotic distribution of the empirical eigenvalues.
\begin{theorem} \label{cor:ev_sqrtn_normality}
 Let \Cref{asu:fixed_dim} be given. 
 \begin{enumerate}[(a)]
     \item Then as $\ninf$,
     \begin{align*}
         (\well_{n,1}, \ldots, \well_{n,d-1}) = (\lambda_1, \ldots, \lambda_{d-1}) + O_\P(1/ \sqrt{k_n}),
     \end{align*}
     \item and \begin{align*}
        \sqrt{k_n} \big( (\well_{n,\ps+1}, \ldots, \well_{n,d-1}) - \lambda \mathbf{1}_{d- \ps -1} \big) \Dconv \bM,
    \end{align*}
  where the entries of the random vector $\bM\in\R^{d-\ps -1}$ are the $(d-\ps -1)$ largest eigenvalues of $\bP_\lambda \bS \bP_\lambda$ in decreasing order, $\bS$ is defined as in \Cref{th:asymp_theta_cov} and $\bP_\lambda \in \R^{d \times d}$ is the orthogonal projection onto the space spanned by the eigenvectors with respect to the eigenvalue $\lambda$ of $\bSigma$.

 \end{enumerate}
 
\end{theorem}

\subsection{Directional model in the high-dimensional case} \label{sec:high_dim_asymptotic}  \label{sec:lsd}

In the high-dimensional setting, where $d = d_n$ depends on $n$ and $d_n \rightarrow \infty$ as $\ninf$, we restrict our studies to a parametric family of distributions, the so-called \textit{directional model}. A directional model
has the advantage that the underlying random vectors have an independent radial and directional component, but still the covariance matrix $\bSigma^{(n)}$ has a spiked structure. The explicit definition of a directional model is the following.

\medskip
\noindent \textsc{Directional Model (D)}: \hspace{0.2cm}
\textit{Suppose for any $n\in\N$ that }
\begin{align} \label{def:Gamma}
    \sbSigman \coloneqq \begin{pmatrix}
        \bGamman & \boldsymbol{0}_{\ps \times d_n} \\
        \boldsymbol{0}_{d_n \times \ps} & \bI_{d_n -\ps}
    \end{pmatrix} \in \R^{d_n \times d_n},
\end{align} 
\textit{where $\bGamman \in \R^{\ps \times \ps}$ is a covariance matrix  with eigenvalues $$\slambda{1} \ge \cdots \ge \slambda{\ps} > 1,$$  $\bVn=(V_1,\ldots,V_{d_n})^{\top} \in \R^{d_n}$ is a centered random vector consisting of i.i.d. symmetric components with variance $1$ and finite fourth moment and $Z$ is a standard Fréchet distributed random variable. Then the sequence of random vectors $( \bXn)_{n \in \N}$ with 
\begin{align*}
    \bXn \coloneqq \frac{ \phantom{\Vert} \sqrtsbSigman \bVn  \phantom{\Vert}}{\Vert \sqrtsbSigman \bVn \Vert} \cdot Z \in \R^{d_n},
\end{align*} 
follows the so-called directional model.}
\medskip

Due to construction, we see directly that the directional component  
\begin{align*}
   \bThetan \coloneqq \frac{ \phantom{\Vert}  \bXn  \phantom{\Vert}}{\Vert \bXn \Vert } =   \frac{  \phantom{\Vert} \sqrtsbSigman \bVn  \phantom{\Vert}}{\Vert \sqrtsbSigman \bVn \Vert},
\end{align*}
of $\bXn$ is independent of the radial component $\Vert \bXn \Vert=Z$, and additionally, $\bThetan$  is the spectral vector of the multivariate regularly varying random vector $\bXn$  of index $1$. Thus, the dependence structure of $\bXn$ is completely determined by $\bThetan$. 

\begin{remark}~
\begin{enumerate}[(a)]
   \item In high-dimensional models it is necessary to specify the model as we have done with the directional model, because due to the increase in dimensionality, the empirical covariance matrix and even the covariance matrix do not converge and hence, it will be impossible to get any kind of limit results without assuming some structure on the model.
    \item Scaling of $\bVn$ has no influence on the distribution of $\bXn$, therefore setting the variance  of $V_i$ to $1$ is no restriction.
    \item The \textit{empirical spectral distribution} (\citet[p. 5]{Bai:Silverstein:2010}) of $\sbSigman$ is defined as
\begin{align*}
    F^{\sbSigman}(x) = \frac{1}{d_n} \sum_{i=1}^{d_n} \mathbbm{1} \{ \slambda{i} \le  x \}, \quad x\in\R,
\end{align*}
and the \textit{limiting spectral distribution} (LSD) of $\sbSigman$ is the Dirac measure $\delta_1$, since
\begin{align*}
   \limn F^{\sbSigman } (x)  %= \limn \frac{1}{d_n} \sum_{j=1}^{d_n} \mathbbm{1} \{ x \le \sigma_{n,j}  \} \\
   =  \limn \frac{1}{d_n} \sum_{j=1}^{p} \mathbbm{1} \{ \slambda{j}\le   x \} + \frac{d_n -p}{d_n} \mathbbm{1} \{1 \le  x  \} =  \mathbbm{1} \{1 \le  x  \}, \quad x \in \R.
\end{align*}
\end{enumerate}
    
\end{remark}
In the following, we denote 
 the covariance matrix of $\bThetan$ as 
 \begin{equation*}
     \bSigman\coloneqq \Cov(\bThetan)
 \end{equation*}
 whereas $\sbSigman$ is the covariance matrix of the non-standardized directional component $\sqrtsbSigman \bVn$. Not only $\bGamman$  has the eigenvalues $\slambda{1}, \ldots, \slambda{\ps}$ but $\sbSigman$ has likewise these eigenvalues.  Additionally, $\sbSigman$ has $(d_n - {\ps})$-times the eigenvalue $1$ which we denote as well as $\slambda{\ps+1}, \ldots, \slambda{d_n}$.
 We are still in the setup of the last section because not only the eigenvalues of $\sbSigman$  satisfy the spiked covariance structure 
 \begin{equation*}
    \slambda{1}\geq \cdots \geq  \slambda{\ps}>1=\slambda{\ps+1}= \cdots =\slambda{d_n} 
 \end{equation*}
 in \Cref{sec:model} but as well the eigenvalues of $\bSigman$ satisfy
 the structure in \Cref{sec:model} although  $\bSigman$ has different eigenvalue than  $\sbSigman$.

\begin{lemma} \label{lem:hd_spiked_covariance_structure}
   Let the Directional Model (D) be given.
   Then for any $n\in\N$,
   \begin{equation*}
       \lambda_{\ps}(\bSigman) > \lambda_{\ps + 1}(\bSigman) = \cdots = \lambda_{d_n}(\bSigman).
   \end{equation*}
\end{lemma}
Hence, there is a spike after the $\ps$-th eigenvalue of $\bSigman$ and the eigenvalues $\lambda_{\ps +1}(\bSigman),$ $\ldots, \lambda_{d_n - 1}(\bSigman)$ are all equal, as required in the definition of the spiked covariance model in \Cref{sec:model}.
We summarize the model in the following.

\begin{assumptionletter}~ \label{asu:high_dim}
%\vspace*{-0.5cm}
\begin{enumerate}
    \item[(B1)]
Let  $\bXn, \bXn_1, \bXn_2, \ldots, \bXn_n$ be an i.i.d. sequence of $d_n$-dimensional random vectors satisfying the Directional Model (D) with  $\E[|V_1|^4] < \infty$. 
\item[(B2)] The eigenvalues $\slambda{1}, \ldots, \slambda{d_n}$ of $\sbSigman$  in \Cref{def:Gamma} satisfy
\begin{equation*}
\slambda{1} \ge \cdots \ge \slambda{\ps} > 1=\slambda{\ps+1}=\cdots=\slambda{d_n}.
\end{equation*}
    \item[(B3)] Let
 $(k_n)_{n \in \N}$ be a sequence in $\N$ with $k_n \rightarrow \infty$,  $k_n/n \rightarrow 0$ and $$d_n/ k_n \rightarrow c > 0, \quad \text{ as } \ninf.$$ 
\end{enumerate}
\end{assumptionletter}

\begin{remark}~
\begin{enumerate}[(a)]
    \item  The assumption $d_n/ k_n \rightarrow c > 0$ as $\ninf$ guarantees that the dimension $d_n$ increases with a rate similar to the number of extremes $k_n$.
    \item  Eigenvalues which are larger than $1 + \sqrt{c}$, are \textit{called distant spiked eigenvalues}, whereby the asymptotic behavior of the corresponding empirical eigenvalues changes if they are above or below $1 + \sqrt{c}$; see 
 the following theorem. Due to \citet[Theorem 4.1 and Theorem 4.2]{Silverstein:Choi:1995}, the assumption $\slambda{\ps} > 1 + \sqrt{c}$  is equivalent  to $\varphi_c'( \slambda{\ps}) > 0$ where
    \begin{equation}
    \varphi_c(x)  \coloneqq x \left(  1 +  c \int \frac{t}{x - t } \di \delta_{1}(t) \right)= x \left( 1 + \frac{c}{x-1}  \right). \label{eq:def_phi}
\end{equation}   
\end{enumerate}
\end{remark}
Analog to \Cref{def:sigma_estimator} 
we define the $d_n \times d_n$ empirical covariance matrix of $\bSigman$ as
\begin{equation}
     \bhSigman\coloneqq    \frac{1}{k_n}  \sum_{j=1}^n  \left( \frac{\bXn_j}{\Vert \bXn_j \Vert} -   \widehat{\bTheta}{}^{(n)}  \right) \cdot \left( \frac{\bXn_j}{\Vert \bXn_j \Vert} -   \widehat{\bTheta}{}^{(n)}  \right)^\top \mathbbm{1}\{ \Vert \bXn_{i} \Vert  >  \Vert \bXn_{(k_n+1,n)} \Vert  \} , \label{def:sigma_estimator_hd}
\end{equation}
with eigenvalues $\well_{n,1}, \ldots, \well_{n,d_n}$, where 
\begin{equation*}
   \widehat{\bTheta}{}^{(n)}  \coloneqq \frac{1}{k_n} \sum_{i=1}^n  \frac{\bXn_i}{\Vert \bXn_i \Vert} \mathbbm{1}\{ \Vert \bXn_{i} \Vert  >  \Vert \bXn_{(k_n+1,n)} \Vert  \}.
\end{equation*}
In contrast to the empirical covariance matrix $\bhSigma_n $  in \Cref{def:sigma_estimator} with a fixed dimension $d\times d$, the dimension of the empirical covariance matrix $\bhSigman$ in \eqref{def:sigma_estimator_hd} is $d_n\times d_n$ and hence, growing in $n$. 

Let us first present the asymptotic distribution of the eigenvalue $\well_{n,1}, \ldots, \well_{n,d_n}$ of $\bhSigman$ under the constraint that $\bGamman$ and its eigenvalues $\slambda{1}, \ldots, \slambda{\ps}$ are converging, and afterwards when $\slambda{\ps}\to\infty$.

\begin{theorem} \label{lem:Bai_Yaio2}
     Let \Cref{asu:high_dim} be given. Suppose that {$\bGamman \rightarrow \bGamma$} 
     and $(\slambda{1}, \ldots, \slambda{\ps}) \rightarrow (\slambdaohnen_1, \ldots, \slambdaohnen_{\ps})$ as $\ninf$ with $\slambdaohnen_{\ps} > 1 + \sqrt{c}$. 
    \begin{enumerate}[(a)]
        \item Let $i \in\{ 1, \ldots, \ps\}$. Then the asymptotic behavior 
        \begin{equation*}
            d_n \well_{n,i} \Pconv \varphi_c(\slambdaohnen_i),  \quad \text{ as }  \ninf 
        \end{equation*}        
        holds, where $\varphi_c$ is defined as in \Cref{eq:def_phi}.
        \item Let $(i_n(\alpha))_{n \in \N}$ be a sequence in $\N$ with $i_n(\alpha) > \ps$ and $i_n(\alpha)/d_n \rightarrow \alpha \in [0,1]$ for any $\alpha\in(0,1)$. Then  
        \begin{equation*}
           \sup_{\alpha\in(0,1)}\left\vert d_n \well_{n,i_n(\alpha)}-  F^\leftarrow_c(1 - \alpha)\right\vert \Pconv 0,  \quad   \text{ as }  \ninf,
        \end{equation*}
         where $F^\leftarrow_c$ is the generalized inverse of $F_c$  with density 
\begin{equation*}  % \label{eq:MP_law}
    f_c(x) =\begin{cases}
        \frac{1}{2 \pi x c} \sqrt{((1 + \sqrt{c})^2-x) ( x-(1 - \sqrt{c})^2)}, & x \in ((1 - \sqrt{c})^2,(1 + \sqrt{c})^2), \\
        0, & \text{otherwise},
    \end{cases}
\end{equation*}
and point mass $1 - 1/c$ at $0$ if $c > 1$.
In particular, if $(q_n)_{n\in\N}$ is a sequence in $\N$ with $q_n=o(d_n)$ and $q_n>p^*$, then
         $d_n \well_{n,q_n} \Pconv (1+\sqrt{c})^2$.
        \item Suppose $0<c\le1$ and $(q_n)_{n\in\N}$ is a sequence in $\N$ with $q_n=o(d_n)$ as $n\to\infty$.
        Then as $n\to\infty$,
        \begin{equation*}
             \frac{1}{d_n-q_n} \sum_{i=q_n+1}^{d_n} d_n \well_{n,i} \Pconv 1.
\end{equation*}
        \item Suppose $c>1$ and $(q_n)_{n\in\N}$ is a sequence in $\N$ with $q_n=o(d_n)$ as $n\to\infty$.
        Then as $n\to\infty$,
        \begin{equation*}
             \frac{1}{k_n-q_n} \sum_{i=q_n+1}^{k_n} d_n  \well_{n,i} \Pconv   c.
\end{equation*}
    \end{enumerate}
\end{theorem}

\begin{remark}%~
        The limiting spectral distribution $F_c$ is called Mar\v{c}enko-Pastur law after the authors of  \cite{marchenko:pastur:1967} and plays an important role in random matrix theory (cf. \citet{Bai:Silverstein:2010}). \citet{marchenko:pastur:1967} first derived for random matrices with i.i.d. components the asymptotic distribution of the eigenvalues of the empirical covariance matrix when the sample size and the dimension tend to infinity, which differs from the classical statistical setting with fixed dimension. 
\end{remark}

So far we have assumed that the first $\ps$ eigenvalues $\slambda{1}, \ldots, \slambda{\ps}$ of $\sbSigman$  are bounded. Alternatively, it is also possible to suppose that   $\slambda{\ps} \rightarrow \infty$ as $\ninf$. 
\begin{theorem} \label{th:Bai_Lemma2_2_lambda}
        Let \Cref{asu:high_dim} be given. Suppose  $\slambda{\ps}  \rightarrow \infty$ {and  $\slambda{1} = o(d_n^{1/2})$} as $\ninf$. 
    \begin{enumerate}[(a)]
        \item Let $i \in\{ 1, \ldots, \ps\}$. Then the asymptotic behavior  \label{th:Bai_Lemma2_2_lambda_a}
        \begin{equation*}
            d_n \well_{n,i}/\slambda{i} \Pconv 1,  \quad \text{ as }  \ninf 
        \end{equation*}
        holds.
        \item Let $(i_n(\alpha))_{n \in \N}$ be a sequence in $\N$ with $i_n(\alpha) > \ps$ and $i_n(\alpha)/d_n \rightarrow \alpha \in [0,1]$ for any $\alpha\in(0,1)$. Then  
        \begin{equation*}
           \sup_{\alpha\in(0,1)}\left\vert d_n \well_{n,i_n(\alpha)}-  F^\leftarrow_c(1 - \alpha)\right\vert \Pconv 0,  \quad   \text{ as }  \ninf,
        \end{equation*}
         where $F^\leftarrow_c$ is defined as in \Cref{lem:Bai_Yaio2}. 
In particular, if $(q_n)_{n\in\N}$ is a sequence in $\N$ with $q_n=o(d_n)$ and $q_n>p^*$, then
         $d_n \well_{n,q_n} \Pconv (1+\sqrt{c})^2$.
        \item Suppose $0<c\le1$ and $(q_n)_{n\in\N}$ is a sequence in $\N$ with $q_n=o(d_n)$ as $n\to\infty$. \label{th:Bai_Lemma2_2_lambda_b}
        Then as $n\to\infty$,
        \begin{equation*}
             \frac{1}{d_n-q_n} \sum_{i=q_n+1}^{d_n} d_n \well_{n,i} \Pconv 1.
\end{equation*}
        \item Suppose $c>1$ and $(q_n)_{n\in\N}$ is a sequence in $\N$ with $q_n=o(d_n)$ as $n\to\infty$. \label{th:Bai_Lemma2_2_lambda_c}
        Then as $n\to\infty$,
        \begin{equation*}
             \frac{1}{k_n-q_n} \sum_{i=q_n+1}^{k_n} d_n \well_{n,i} \Pconv   c.
\end{equation*}
        \item Suppose $0<c<1$ and let $i \in\{ 1, \ldots, \ps\}$. Then as $n\to\infty$, \label{th:Bai_Lemma2_2_lambda_d}
        \begin{equation*}
             \frac{d_n  \well_{n,i}}{\frac{1}{d_n-i} \sum_{j=i+1}^{d_n} d_n \well_{n,j}} \Pconv \infty.
\end{equation*}
    \end{enumerate}
    
\end{theorem}

\begin{remark}
    The assumption $\slambda{1} = o(d_n^{1/2})$ as $n\to\infty$ guarantees that the largest eigenvalue grows sufficiently slowly compared to the dimension $d_n$. When all moments of $V_1$ exist this assumption can be relaxed to $\slambda{1} = o(d_n^\beta)$ as $n\to\infty$ for any $\beta < 1$ due to \Cref{rem:order_slambda}.
\end{remark}

\section{Information criteria for the number of principal components in the fixed-dimensional case} \label{sec:fixed_dim}

The aim of the paper is to derive estimators for $\ps$, the location of the spike in the eigenvalues of $\bSigma=\Cov(\Theta)$, which defines the dimensionality of the PCA, by exploiting information criteria. In the context of PCA for Gaussian data, an Akaike information criteria (AIC) and a Bayesian information criteria (BIC) was developed in \citet{fujikoshi:sakurai:2016} and the consistency in the high-dimensional case for general data was analyzed in \citet{Bai:Choi:Fujikoshi:2018}. 
The $\AIC$ (\citet{akaike}) is based on minimizing the Kullback-Leibler divergence between the true distribution and the model, and the $\BIC$ (\citet{schwarz}) maximizes the posterior probability. 
In this paper, we adopt these information criteria and study their statistical properties. We start in this section with the fixed-dimensional case and give the proper definitions of the information criteria under \Cref{asu:fixed_dim}.  The proofs of this section are moved to \Cref{sec:fixed_dim_proof}.

\begin{definition} \label{def:aic} Suppose $\well_{n,1}, \ldots, \well_{n,d-1}$ are the empirical eigenvalues of $\bhSigma_n$ as defined in \Cref{def:sigma_estimator}.
\begin{enumerate}[(a)]
    \item The \textit{Akaike information criterion} ($\AIC$) for the fixed-dimensional case is defined as
\begin{align*} 
    \AIC_{k_n}  (p) \coloneqq &  {k_n} \sum_{i=1}^{p} \log  ( \well_{n,i}  )  + {k_n} (d - 1 - p)  \log \big( \frac{1}{d - 1-p} \sum_{j = p+1}^{d - 1} \well_{n,j} \big) \\
    &   + k_n  \log \Big( \frac{k_n -1}{k_n} \Big)^{d-1} + k_n (d-1) ( \log(2 \pi) + 1) \\
    &+ 2(p+1) (d - p/2 ), \;  
\end{align*}
for $p = 1, \ldots, d-2$ and an estimator for $\ps$ is $\widehat{p}_n \coloneqq \argmin_{1 \le p \le d-2} \AIC_{k_n}(p)$.
   \item The \textit{Bayesian information criterion} ($\BIC$) for the fixed-dimensional case is defined as
\begin{align*}
    \BIC_{k_n}  (p) \coloneqq& \, {k_n} \sum_{i=1}^{p} \log  ( \well_{n,i}  )  + {k_n} (d - 1 - p)  \log \big( \frac{1}{d - 1-p} \sum_{j = p+1}^{d - 1} \well_{n,j} \big) \\
    &  + k_n  \log \left( \frac{k_n -1}{k_n} \right)^{d-1} + k_n (d-1) ( \log(2 \pi) + 1) \\
    &+ \log(k_n) (p+1) (d -  \frac{p}{2}  /2),
\end{align*}
for $p = 1, \ldots, d-2$
and an estimator for $\ps$ is $\widehat{p}_n \coloneqq \argmin_{1 \le p \le d-2} \BIC_{k_n}(p)$.
\end{enumerate}
\end{definition}

\begin{remark}\label{rem:scale_invariant}~
\begin{enumerate}[(a)]
    \item 
     The penalty term $2(p+1) (d - p/2 )=(p+1) d - p (1+p)/2$ arises as it is the number of parameters that define a $(d-1)$-dimensional normal distribution with an arbitrary mean vector and covariance matrix following the $p$-th spiked covariance model  (cf. \citet[Section 2]{fujikoshi:sakurai:2016}). 
    As baseline model, we take a $(d-1)$-dimensional normal distribution instead of a $d$-dimensional distribution because  
    $\bTheta$ is a random vector on the unit sphere and hence the first $(d-1)$ components already determine the last component.  In summary, we use a modified version of the $\AIC$ and the $\BIC$ of \citet{fujikoshi:sakurai:2016} by replacing $d$ with $d-1$ and dropping the last empirical eigenvalue $\well_{n,d_n}$.

   \item The $\AIC$ and $\BIC$ are invariant to scaling of the eigenvalues. Consequently, scaling the sample covariance matrix $\bhSigma_n$, or equivalently the eigenvalues $\well_{n,1}, \ldots, \well_{n,d-1}$, does not affect the point at which the information criteria achieve their minimum.
    \end{enumerate}
\end{remark}

Next, we check the consistency of the $\AIC$ and the $\BIC$. First, we present the result for the AIC and second for the BIC.
 
\begin{theorem} \label{th:AIC_Cons}
Let \Cref{asu:fixed_dim} be given and $\bM$  be the limit vector in \Cref{cor:ev_sqrtn_normality}. Then
\begin{equation*}
    \limn \P( \AIC_{k_n}(p) > \AIC_{k_n}(\ps) ) =\begin{cases}
\P( g_p(\bM) >  0 )   &\text{ for } p > {\ps}, \\
 1  &\text{ for } p < {\ps},
\end{cases}
\end{equation*}
where \begin{align*}
    g_p(\bm) &\coloneqq  - \frac{1}{2}  \sum_{i=\ps+1}^{p}  m_{i}^2 - \frac{1}{2(d - 1-p)} \left(\sum_{j = p+1}^{d - 1}  m_{j} \right)^2    + \frac{1}{2(d - 1-\ps)} \left(\sum_{j = \ps+1}^{d - 1}  m_{j} \right)^2  \nonumber  \\
    & \qquad -  (d-p-2)(d-p+1) +   (d-\ps-2)(d-\ps+1)
\end{align*}
for $\bm = (m_1, \ldots, m_d) \in \R^{d}$.

\end{theorem}

\begin{remark}
    Under some technical assumptions on the distribution of $\bTheta$,  it is possible to state a density for $\bM$ (cf. \citet{davis:1977}) and derive that $\P( g_p(\bM) >  0)<1$. For the present paper, it is sufficient to give an example such that the $\AIC$ is not consistent.  
\end{remark}

\begin{example}
 We assume the Directional Model (D), where $\sbSigman:=\bGamma: = \diag(9,4,4,$ $1)$, $\bVn \coloneqq \bV:=(V_{1}, V_{2},V_{3},V_{4}  )^{\top} $ with $V_{i} \sim \mathcal{U}(\{-1, 1  \}),$ $i = 1, \ldots, 4$, $Z \sim \text{Fréchet}(1)$ and the dimension $d = 4$ is fixed. Then, we have $\Vert \sqrtsbSigman \bVn \Vert =\Vert\bGamma\bV \Vert= \sqrt{9 + 4 + 4 +1} = \sqrt{18}$ and $\bThetan:=\bTheta: = (3V_{1},$ $2V_{2},$ $2V_{3},$ $V_{4}  )/\sqrt{18}$. 
 We have $\E[\bTheta] = \mathbf{0}_4,$  $\bSigma = \bGamma/18$, where the eigenvalues of $\bSigma$ are $(1/2,$ $2/9,$ $2/9,$ $1/18)$ and the corresponding eigenvectors are the unit vectors $\be_1, \ldots, \be_4 \in \R^4$. Consequently, the spiked covariance assumption is satisfied with $\lambda = 2/9, d = 4$ and $\ps = 1$. 
 
In the following, we calculate the probability $\P(g_2(\bM) <  0)$ by first determining the asymptotic distribution of $\bM$.
An application of \Cref{cor:ev_sqrtn_normality} (b) yields 
 \begin{equation*}
     \sqrt{k_n} \big( (\well_{n,2},  \well_{n,3}) - (\lambda,  \lambda) \big) \Dconv (M_2,M_3)
 \end{equation*}
 in $\R^{2}$ is the joint distribution of the decreasingly ordered non-null eigenvalues of 
 \begin{equation*}
     \bP_\lambda \bS \bP_\lambda =  \begin{pmatrix}
      0 & 0 & 0 & 0 \\
      0 & 1 & 0 & 0 \\
      0 & 0 & 1 & 0 \\
      0 & 0 & 0 & 0 \\
  \end{pmatrix}
    \begin{pmatrix}
     S_{11}  & S_{12}  & S_{13}  & S_{14} \\
        S_{12}  & S_{22}  & S_{23}  & S_{24} \\
       S_{13}  & S_{23}  & S_{33}  & S_{34} \\
       S_{14}  & S_{24}  & S_{34}  & S_{44} \\
  \end{pmatrix}
    \begin{pmatrix}
       0 & 0 & 0 & 0 \\
       0 & 1 & 0 & 0 \\
       0 & 0 & 1 & 0 \\
       0 & 0 & 0 & 0 \\
  \end{pmatrix}
  = \begin{pmatrix}
      0 & 0 & 0 & 0 \\
       0  & S_{22}  & S_{23}  & 0 \\
       0  & S_{23}  & S_{33}  & 0 \\
        0 & 0 & 0 & 0 \\
  \end{pmatrix},
 \end{equation*}
 where $\ve(\bS)$ follows a centered multivariate normal distribution with covariance $\Cov ( \ve( $ $(\bTheta - \E[\bTheta])(\bTheta - \E[\bTheta])^\top  ) )$ and $\bP_\lambda  \coloneqq ( \be_2, \be_3) \cdot  ( \be_2, \be_3)^\top  \in \R^{4 \times 4}$ is the projection onto the $2$-dimensional eigenspace of the orthonormal eigenvectors $\be_2, \be_3$ corresponding to $\lambda = 2/9$.
Since $\V(S_{22}) = \E[ \Theta_2^4] - (\E[ \Theta_2^2])^2 = 0$ and $\V(S_{33}) = 0$, the distributions of $S_{22}$ and $S_{33}$ are  degenerate with expectation zero. By the symmetry of $\bP_\lambda \bS \bP_\lambda$, the non-null eigenvalues of the matrix $ \bP_\lambda \bS \bP_\lambda$ can be calculated  directly and are given by
\begin{equation*}
    M_2 =  S_{23} \qquad \text{ and }  \qquad  M_3 = -S_{23}. 
\end{equation*}
Next, since  $(d-\ps-2)(d-\ps+1) -  (d-p-2)(d-p+1) = 4$ for $p = 2$ and $\ps = 1$, the inequality $g_2(\bM) <  0$ is equivalent to 
\begin{align*}
   4 < \frac12 M_2^2 + \frac12 M_3^2 - \frac14 (M_2 + M_3)^2 = S_{23}^2 . %\label{eq:AIC_example_condition}
\end{align*}
 Due to the definition of $\bS$, the distribution of $S_{23}$ is Gaussian with expectation zero and $\V(S_{23}) = \E[  \Theta_2^2  \Theta_3^2] = 1 $ so that   $\P(g_2(\bM) <  0)>0$. 
 % Indeed, $S_{23}$ is non-degenerate because
% the random vector $\ve( \bP_\lambda \bS \bP_\lambda ) $ follows a centered normal distribution with covariance matrix 
% \begin{align*}
    % \Cov( \ve(  \bP_\lambda \bTheta  \bTheta^\top  \bP_\lambda  )  ) 
  % &=  \Cov \left( \ve      \begin{pmatrix}
%       0 & 0 & 0 & 0 \\
%       0 & 1 & \Theta_2 \Theta_3  & 0 \\
%       0 & \Theta_2 \Theta_3  & 1 & 0 \\
%       0 & 0 & 0 & 0 \\
%   \end{pmatrix} \right)  \nonumber \\
%   &=  \diag \left( 
%       \mathbf{0}_{6 \times 6} ,
%       \begin{pmatrix}
%        & 1 & 0 &  0 & 1  \\
%       & 0 & 0 & 0 & 0 \\
%       &0 & 0 & 0 & 0 \\
%       & 1 & 0 &  0 & 1  \end{pmatrix},   \mathbf{0}_{6 \times 6} \right) \in \R^{16 \times 16} , \label{eq:cov_example}
% \end{align*}
% where we used that $\E[\bTheta] = \mathbf{0}_4,$  $\Theta_i^2 = 1, i = 1, \ldots, 4$ and $\Cov( \Theta_2  \Theta_3, \Theta_2  \Theta_3) = \E[  \Theta_2^2  \Theta_3^2] = 1 $. For clarity, the last matrix has been written as a block diagonal matrix. \hfill$\Box$

\end{example}

In contrast to the AIC, the $\BIC$ is a weakly consistent information criterion and selects the true dimension $\ps$ with probability converging to $1$ as stated in the next theorem.

\begin{theorem} \label{th:BIC_Cons}
Let \Cref{asu:fixed_dim}  be given. Then
\begin{equation*}
    \limn \P( \BIC_{k_n}(p) > \BIC_{k_n}(\ps) ) = 1 \quad  \text{ for } p \ne {\ps}.
\end{equation*}
\end{theorem}

The non-consistency of the AIC and the consistency of the BIC are typical for these information criteria in the fixed-dimensional case. In the high-dimensional case, the asymptotic properties, derived in the next section, differ.

\section{Information criteria for the number of principal components in the high-dimensional case}  \label{sec:high_dim}

The topic in this section is information criteria in the high-dimensional case of \Cref{asu:high_dim}, where $d = d_n$ depends on $n$ and $d_n / k_n \rightarrow  c > 0$ as $\ninf$. For the definition of the information criteria and the asymptotic properties, we need to differentiate between the cases $c < 1$ and $c > 1$. The reason behind it is that if $d_n > k_n$, the last  $d_n - k_n$ empirical eigenvalues of $\bhSigman$ are equal to zero, i.e. $\well_{n,k_n+1}= \cdots= \well_{n,d_n} = 0$. Therefore, in \Cref{sec:high_dim_d_klein}, we analyze the information criteria for $0 < c < 1$ and in \Cref{sec:high_dim_d_groß} for $c > 1$. 
The proofs of this section are provided in \Cref{{sec:proofs:C}}.

\subsection{Information criteria for  \texorpdfstring{$0 < c < 1$}{0 < c < 1}} \label{sec:high_dim_d_klein}
In the case $0 < c < 1$, the definition of the information criteria are similar to the fixed-dimensional setting but we would like to point out that in the high-dimensional setting, we do not necessarily evaluate the information criteria at all possible values $1, \ldots, d_n-1$ but rather restrict to $1, \ldots, q_n$ with $q_n \le d_n$. The number $q_n$ is called the number of candidate dimensions. 

\begin{definition} \label{def:aic_high_dim_c_klein} Suppose $\well_{n,1}, \ldots, \well_{n,d_n-1}$ are the empirical eigenvalues of $\bhSigman$ as defined in \Cref{def:sigma_estimator_hd} and let $q_n \le d_n-2$.
\begin{enumerate}[(a)]
    \item The \textit{Akaike information criterion} ($\QAIC$) for the high-dimensional case with $d_n < k_n$ is defined as
\begin{align*}
      \QAIC_{k_n}  (p) \coloneqq&  \, \sum_{i=1}^{p} \log  ( \well_{n,i}  )  +  (d_n - 1 - p)  \log \big( \frac{1}{d_n - 1-p} \sum_{j = p+1}^{d_n - 1} \well_{n,j} \big) \\
    &  + \log \left( \frac{k_n -1}{k_n} \right)^{d_n-1} +  (d_n-1) ( \log(2 \pi) + 1)+  \frac{ (p+1) (2d_n - p)}{k_n} ,
\end{align*}
for $p = 1, \ldots,d_n -2$ and an estimator for $\ps$ is $\widehat{p}_n \coloneqq \argmin_{1 \le p \le q_n} \QAIC_{k_n}(p)$.
\item The \textit{Bayesian information criterion} ($\QBIC$) for the high-dimensional case with $d_n < k_n$ is  defined as
\begin{align*}
          \QBIC_{k_n}  (p) \coloneqq&  \, \sum_{i=1}^{p} \log  ( \well_{n,i}  )  +  (d_n - 1 - p)  \log \big( \frac{1}{d_n - 1-p} \sum_{j = p+1}^{d_n - 1} \well_{n,j} \big) \\
    &  + \log \left( \frac{k_n -1}{k_n} \right)^{d_n-1} +  (d_n-1) ( \log(2 \pi) + 1) \\
    &  +\log(k_n) \frac{(p+1) (d_n -  p/2)}{k_n} ,
\end{align*}
for $p = 1, \ldots, d_n -2$ and an estimator for $\ps$ is $\widehat{p}_n \coloneqq \argmin_{1 \le p \le q_n} \QBIC_{k_n}(p)$.
\end{enumerate}
\end{definition}

In the next theorem, we present sufficient assumptions for the $\QAIC$ to be weakly consistent, i.e.,
\begin{align*}
            \limn  \P\Big(   \argmin_{1 \le p < q_n }  \QAIC_{k_n}(p)  = \ps \Big) = 1
        \end{align*}
and afterwards for the $\QBIC$.
 
\begin{theorem} \label{th:cons_aic_high_dim}  Let \Cref{asu:high_dim}  with $0 < c < 1$ be given and let the number $q_n$ of candidate dimensions  satisfy $q_n = o(d_n)$ as $n\to\infty$. 
    \begin{enumerate}[(a)]
        \item Suppose $\bGamman \rightarrow \bGamma$ and $(\slambda{1}, \ldots, \slambda{\ps}) \rightarrow (\slambdaohnen_1, \ldots, \slambdaohnen_{\ps})$ as $\ninf$ with $\slambdaohnen_{\ps} > 1 + \sqrt{c}$. If the gap condition 
        \begin{align}
    \varphi_c( \slambdaohnen_{\ps}) - 1 - \log( \varphi_c( \slambdaohnen_{\ps})) - 2 c > 0 \label{cond:gap}
\end{align}
with $\varphi_c$ as defined  in \Cref{eq:def_phi} holds, then the $\QAIC$ is weakly consistent.
        \item Suppose $\bGamman \rightarrow \bGamma$ and $(\slambda{1}, \ldots, \slambda{\ps}) \rightarrow (\slambdaohnen_1, \ldots, \slambdaohnen_{\ps})$ as $\ninf$ with $\slambdaohnen_{\ps} > 1 + \sqrt{c}$. If the gap condition \Cref{cond:gap} does not hold,  then
        \begin{align*}
            \limn \P\Big(   \min_{1 \le p < {\ps} }  \Big\{ \QAIC_{k_n}(p) - \QAIC_{k_n}(\ps) \Big\} > 0\Big) < 1 
         \end{align*}
         and the $\QAIC$ is not weakly consistent.
       \item Suppose $\slambda{\ps}  \rightarrow \infty$ and $\slambda{1} = o(d_n^{1/2})$ as $\ninf$. Then the $\QAIC$ is weakly consistent.
    \end{enumerate}
    
\end{theorem}

\begin{remark}~
    \begin{enumerate}[(a)]
        \item The division of $\QAIC$ by $k_n$ in contrast to the AIC has no influence in applications, as it does not affect the location of the minimum of the information criteria for a fixed sample size $n$. As a result, in the simulation study, the minima of $\AIC$ and $\QAIC$ coincide, and we do not need to distinguish between these criteria. The  division by $k_n$ in the definition of $\QAIC$, as  in \citet{Bai:Choi:Fujikoshi:2018}, ensures that the limit of the information criteria exists.
        \item The gap condition \Cref{cond:gap} was introduced in \citet{Bai:Choi:Fujikoshi:2018} and it also guarantees that the gap between $\slambdaohnen_{\ps}$ and the non-dominant eigenvalues is sufficiently large.
    \end{enumerate}
\end{remark}

In the following theorem, consistency criteria for the $\QBIC$ are stated, which are slightly different from the results for the $\QAIC$.

\begin{theorem} \label{th:cons_bic_high_dim}
Let \Cref{asu:high_dim} with $0 < c < 1$ be given. Suppose that either $$\bGamman \rightarrow \bGamma \text{ such that } (\slambda{1}, \ldots, \slambda{\ps}) \rightarrow (\slambdaohnen_1, \ldots, \slambdaohnen_{\ps})  \text{ as } n\to \infty \text{ with }\slambdaohnen_{\ps} > 1 + \sqrt{c},$$ 
or $$\slambda{\ps} \rightarrow \infty \quad \text{ and } \quad \slambda{1} = o(d_n^{1/2}) \quad\text{ as } n\to\infty.$$ 
    \begin{enumerate}[(a)]
        \item If $\slambda{\ps} / \log(d_n) \rightarrow 0$ as $\ninf$, then 
        \begin{align*} 
            \limn \P\Big(   \min_{1 \le p < {\ps} }  \Big\{ \QBIC_{k_n}(p) - \QBIC_{k_n}(\ps) \Big\} > 0\Big) < 1  
         \end{align*}
         and the $\QBIC$ is not weakly consistent.
        \item If $\slambda{\ps} / \log(d_n) \rightarrow \infty$ as $\ninf$, then the $\QBIC$ is weakly consistent.
    \end{enumerate}
\end{theorem}

\begin{remark}~
    \begin{enumerate}[(a)]
   
        \item When the gap condition is fulfilled,  the $\QAIC$ is weakly consistent whereas the consistency of the $\QBIC$ depends on the properties of $\slambda{\ps}$.  The $\QBIC$ 
        and, if the gap condition is violated, the $\QAIC$, tends to underestimate the number of significant principal components. A similar result was also obtained by \citet{MVT_regression_high_dim} for multivariate linear regressions in high dimensions. 
        \item The consistency of the $\QAIC$ and $\QBIC$ in the high-dimensional case is opposite to the fixed-dimensional case. Specifically, while the $\AIC$ may not be consistent and the $\BIC$ is consistent in the fixed-dimensional setting, the opposite behavior is observed in the high-dimensional setting. Moreover, in \Cref{th:cons_aic_high_dim} (b), we have
      \begin{align*}
            \limn \P\Big(   \min_{1 \le p < {\ps} }  \Big\{ \QAIC_{k_n}(p) - \QAIC_{k_n}(\ps) \Big\} > 0\Big) < 1 , 
         \end{align*}
        which is opposite to the fixed-dimensional case, where the $\AIC$ tends to overestimate rather than underestimate the number of principal components.  
        \item  The case  $ c = 1$ is excluded from the consideration due to potential complications with the asymptotic behavior of the eigenvalues (see \citet[Section 4]{Bai:Choi:Fujikoshi:2018}). While \Cref{lem:Bai_Yaio2} and \Cref{th:Bai_Lemma2_2_lambda} are valid for $c = 1$, issues arise with the convergence of ratios of quantiles of the Mar\v{c}enko-Pastur law in \citet[Lemma 2.3]{Bai:Choi:Fujikoshi:2018}) when $ q_n = o(d_n) $ is not assumed. If   $q_n = o(d_n)$  is assumed, then the results for $ 0 < c < 1 $ also apply to $ c = 1 $. Additionally, the support of the Mar\v{c}enko-Pastur law for $ c = 1 $ is given by the interval $ (0, 4) $, which can lead to empirical eigenvalues close to zero, causing numerical problems when calculating the logarithm of the empirical eigenvalues.
        \item If $\limn \slambda{\ps} / \log(d_n) \in (0,\infty)$ further assumptions are needed to assess the consistency of the $\QBIC$.
    \end{enumerate}
\end{remark}

\subsection{Information criteria for  \texorpdfstring{$c > 1$}{c > 1}} \label{sec:high_dim_d_groß}
For the case $c > 1$ we have to adapt the information criteria. Therefore, we follow the definition of the $\AIC$ and the $\BIC$ in \citet{Bai:Choi:Fujikoshi:2018}, which leads to the following definition.

\begin{definition} \label{def:aic_high_dim} Suppose $\well_{n,1}, \ldots, \well_{n,d_n-1}$ are the empirical eigenvalues of $\bhSigman$ as defined in \Cref{def:sigma_estimator_hd} and let $q_n \le k_n-2$.
\begin{enumerate}[(a)]
    \item The \textit{Akaike information criterion} ($\QAICstar$) for the high-dimensional case with $d_n > k_n$ is defined as
\begin{align*}
        \QAICstar_{k_n}  (p) \coloneqq&  \, \sum_{i=1}^{p} \log  ( \well_{n,i}  )  +  (k_n - 1 - p)  \log \big( \frac{1}{k_n - 1-p} \sum_{j = p+1}^{k_n - 1} \well_{n,j} \big) \\
    & + \log \left( \frac{d_n -1}{d_n} \right)^{k_n-1} +  (k_n-1) ( \log(2 \pi) + 1)+ \frac{ (p+1)(2 k_n - p)}{d_n} ,
\end{align*}
for $p = 1, \ldots,k_n -2$ and an estimator for $\ps$ is $\widehat{p}_n \coloneqq \argmin_{1 \le p \le q_n} \QAICstar_{k_n}(p)$.
\item The \textit{Bayesian information criterion} ($\QBICstar$) for the high-dimensional case with $d_n > k_n$ is  defined as
\begin{align*}
    \QBICstar_{k_n}  (p) \coloneqq&   \sum_{i=1}^{p} \log  ( \well_{n,i}  )  +  (k_n - 1 - p)  \log \big( \frac{1}{k_n - 1-p} \sum_{j = p+1}^{k_n - 1} \well_{n,j} \big) \\
    &   + \log \left( \frac{d_n -1}{d_n} \right)^{k_n-1} +  (k_n-1) ( \log(2 \pi) + 1)  \\
    &+ \log(d_n) \frac{(p+1) (k_n -  p/2)}{d_n},
\end{align*}
for $p = 1, \ldots, k_n -2$ and an estimator for $\ps$ is $\widehat{p}_n \coloneqq \argmin_{1 \le p \le q_n} \QBICstar_{k_n}(p)$.
\end{enumerate}
\end{definition}

For the consistency analysis of the $\QAICstar$ and $\QBICstar$ we use the same definition for weakly consistent as for the $\QAIC$ in \Cref{sec:high_dim_d_klein}.

\begin{theorem}  \label{th:cons_high_dim_greater_1}
    Let \Cref{asu:high_dim} with $c > 1$ be given and let the number $q_n$ of candidate dimensions satisfy $q_n = o(d_n)$ as $n\to\infty$. 
    \begin{enumerate}[(a)]
        \item Suppose $\bGamman \rightarrow \bGamma$ and $(\slambda{1}, \ldots, \slambda{\ps}) \rightarrow (\slambdaohnen_1, \ldots, \slambdaohnen_{\ps})$ as $\ninf$ with $\slambdaohnen_{\ps} > 1 + \sqrt{c}$. If the modified gap condition \begin{align}
    \frac{\varphi_c( \slambda{\ps})}{c} - 1 - \log \left(  \frac{\varphi_c( \slambda{\ps})}{c} \right) - \frac{2}{c} > 0 \label{cond:gap_high_dim}
\end{align} with $\varphi_c$ as defined  in \Cref{eq:def_phi} holds, then the $\QAICstar$ is weakly consistent.
        \item Suppose $\bGamman \rightarrow \bGamma$ and $(\slambda{1}, \ldots, \slambda{\ps}) \rightarrow (\slambdaohnen_1, \ldots, \slambdaohnen_{\ps})$ as $\ninf$ with $\slambdaohnen_{\ps} > 1 + \sqrt{c}$. If the modified gap condition \Cref{cond:gap_high_dim} does not hold,  then the $\QAICstar$ is not weakly consistent.
       \item Suppose that $\slambda{\ps}  \rightarrow \infty$ and $\slambda{1} = o(d_n^{1/2})$ as $\ninf$. Then the $\QAICstar$ is weakly consistent.
    \end{enumerate}
\end{theorem}

\begin{theorem} \label{th:cons_bic_high_dim_greater_1}
  Let \Cref{asu:high_dim}  with $ c > 1$ be given.
  Suppose that either $$\bGamman \rightarrow \bGamma \text{ such that } (\slambda{1}, \ldots, \slambda{\ps}) \rightarrow (\slambdaohnen_1, \ldots, \slambdaohnen_{\ps})  \text{ as } n\to \infty \text{ with }\slambdaohnen_{\ps} > 1 + \sqrt{c},$$ 
or $$\slambda{\ps} \rightarrow \infty \quad \text{ and } \quad \slambda{1} = o(d_n^{1/2}) \quad\text{ as } n\to\infty.$$ 
        \begin{enumerate}[(a)]
        \item If $\slambda{\ps} / \log(d_n) \rightarrow 0$ as $\ninf$, then the $\QBICstar$ is not weakly consistent.
        \item If $\slambda{\ps} / \log(d_n) \rightarrow \infty$ as $\ninf$, then the $\QBICstar$ is weakly consistent.
    \end{enumerate}
\end{theorem}

\begin{remark}
    \begin{enumerate}[(a)]~
        \item The $\QAICstar$ is weakly consistent when the gap condition is fulfilled and not consistent otherwise, whereas the consistency of the $\QBICstar$ depends on the asymptotic behavior of $\slambda{\ps}$. The results are identical to the case $0 < c < 1$.
        \item Since the last $(d_n - k_n)$ eigenvalues of $\bhSigman$ are equal to $0$,  additional simulation studies showed that if the dimension $d_n$ is sufficiently large, setting some \ eigenvalues of $\bSigman$ to zero has no big influence on the performance of the $\QAICstar$ and $\QBICstar$. However, when $c < 1$, the zero eigenvalues do influence the performance of the $\AIC$ and $\BIC$. In such cases, we recommend first projecting the data onto a lower-dimensional space to ensure that the zero eigenvalues have no impact on the analysis.
    \end{enumerate}
\end{remark}

\section{Simulation study} \label{sec:simulation}
In this section, we compare the performance of the different information criteria through a simulation study.  In the following, we simulate $n$ times a multivariate regularly varying random vector $\bX$ of dimension $d_n$. For the distribution of $\bX$, we distinguish three models. First, in \Cref{sec:simulation_directional} we use the directional model and in \Cref{sec:simulation_directional_noise}, we extend the directional model by adding an additional noise term. 
Finally, the model in \Cref{sec:AGM} exhibits asymptotic dependence but differs from the directional model. In all models, we estimate the parameter $\ps$ by $\widehat{p}_n$ based on $n$ observations. We run the simulations with $500$ repetitions. Throughout these examples,  $c = d_n / k_n$. When $c < 1$ we use the $\AIC$ and the $\BIC$, and if $c > 1$ we use the $\QAICstar$ and the $\QBICstar$.  If for some $c$, $k_n$ is  larger than $n$, we set $k_n = n$.  The code for the simulations is available at \url{https://gitlab.kit.edu/projects/178647}.

\subsection{Directional model} \label{sec:simulation_directional}
First, we consider the Directional Model (D) with $\ps = 10$ as introduced in \Cref{sec:high_dim_asymptotic}. On the one hand, we investigate the fixed-dimensional case with $d = 20$ and on the other hand, the high-dimensional case with $d = 100,$ $200$ and $300$. For comparison, we run simulations with sample sizes $n = 1000,$ $5000,$ $10000$. The matrix $\bGamman$ from \Cref{def:Gamma} is fixed and the eigenvalues $\slambda{1}, \ldots, \slambda{{\ps}}$ are all equal to $\lambda^*$, which is chosen to be larger than 1 and to satisfy the distant spiked eigenvalue condition $\lambda^*>1+\sqrt{c}$.  The entries of $\bVn$ are i.i.d. standard normally distributed.

\begin{figure}[ht]
    \centering
    \includegraphics[width=1\textwidth]{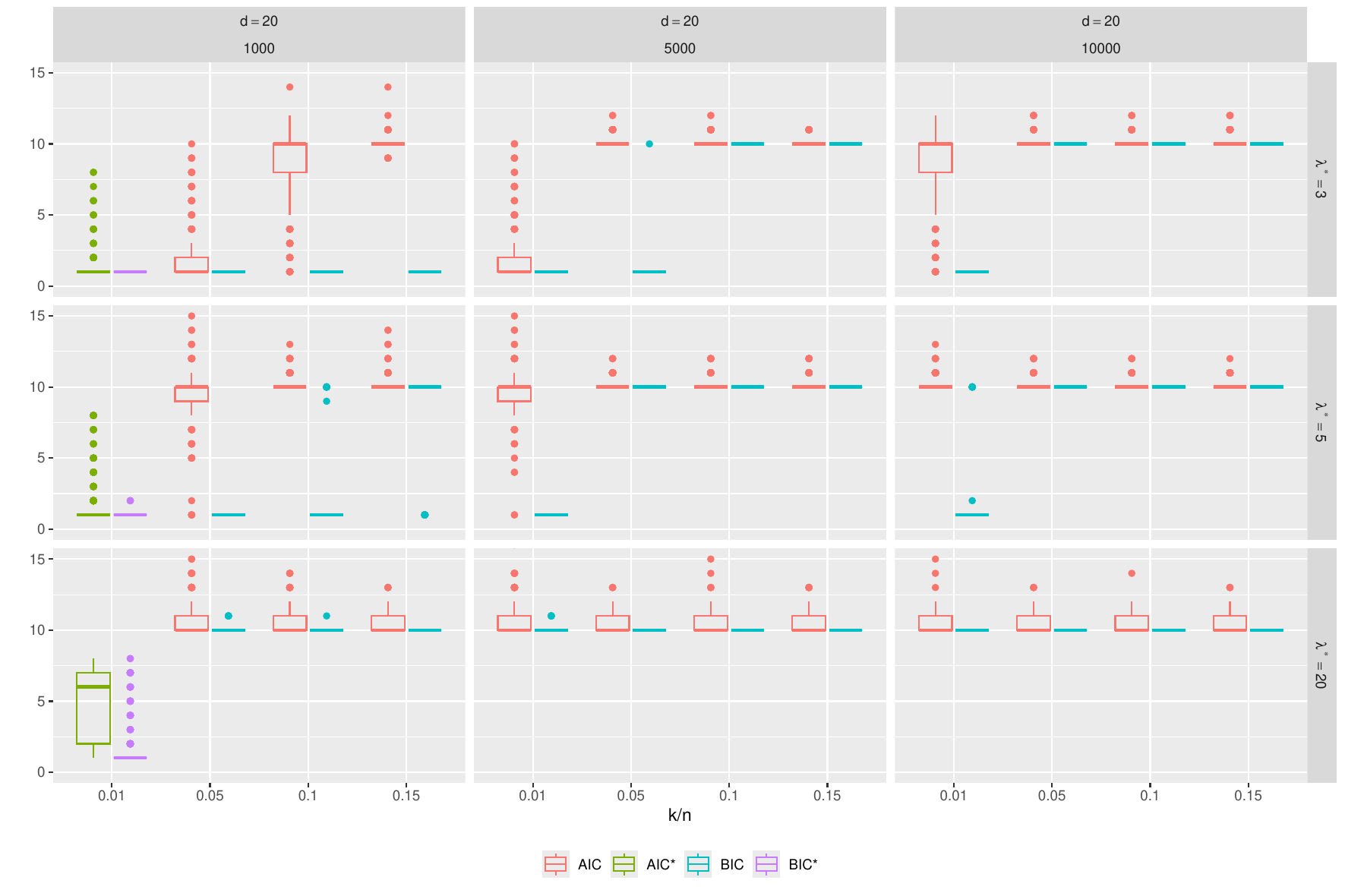}
    \caption{ Simulations for the directional model with $\ps = 10$ and dimension $d = 20$: From left to the right the  sample size increases from $ n = 1000$, $n = 5000$ to $n = 10000$. From top to bottom, the value of the relevant eigenvalues increases from $ \lambda^* = 3$, $\lambda^* = 5$ to $\lambda^* = 20$. In every subplot the ratio $k_n / n$ increases from left to right from  $0.01$, $ 0.05$, $0.1$ to $ 0.15$.  The box plots show the estimator $\widehat p_n$  for $\ps=10$ of the $\AIC$ and $\BIC$.}
    \label{fig:DFM}
\end{figure}

The results for $d = 20$  are presented in \Cref{fig:DFM}. The estimator $\widehat p_n$ of both information criteria gets closer to the true value $\ps = 10$ if $k_n$ increases. For $n = 1000$ and $k_n/n = 0.01$, we have $k_n = 10 < d=20$ and therefore we use the $\QAICstar$ and $\QBICstar$. Both information criteria underestimate $\ps$, which is expected as the number of extreme observations $k_n$ equals $\ps$. In all other cases, the $\AIC$ and $\BIC$ are used. %For $n=5000,10000$ and $k_n/n = 0.01$, the $\AIC$ also underestimates $\ps$ for $\lambda^* \le 3$, which subsides for larger $\lambda^*$.
For $k_n/n \ge 0.05$ and $\lambda^* = 3$, the $\AIC$ either estimates $\ps$ or shows more outliers above $\ps$. 
Overall, the $\AIC$ performs better, when $\lambda^*$ or $k_n$ increases. The $\BIC$ estimates the true value of $\ps$ or underestimates $\ps$, where the number of cases with underestimation becomes smaller when  $\lambda^*$ or $k_n$ grows. This is also intuitive: for a higher value of $\lambda^*$  the spike is more pronounced. % For $\lambda^*=3$ we are close to $1+\sqrt{c}$, so good estimation results  can not be expected.   
In comparison to the $\AIC$, the estimates of the $\BIC$ have, in general, fewer fluctuations and outliers.

\begin{figure}[ht]
    \centering
    \includegraphics[width=1\textwidth]{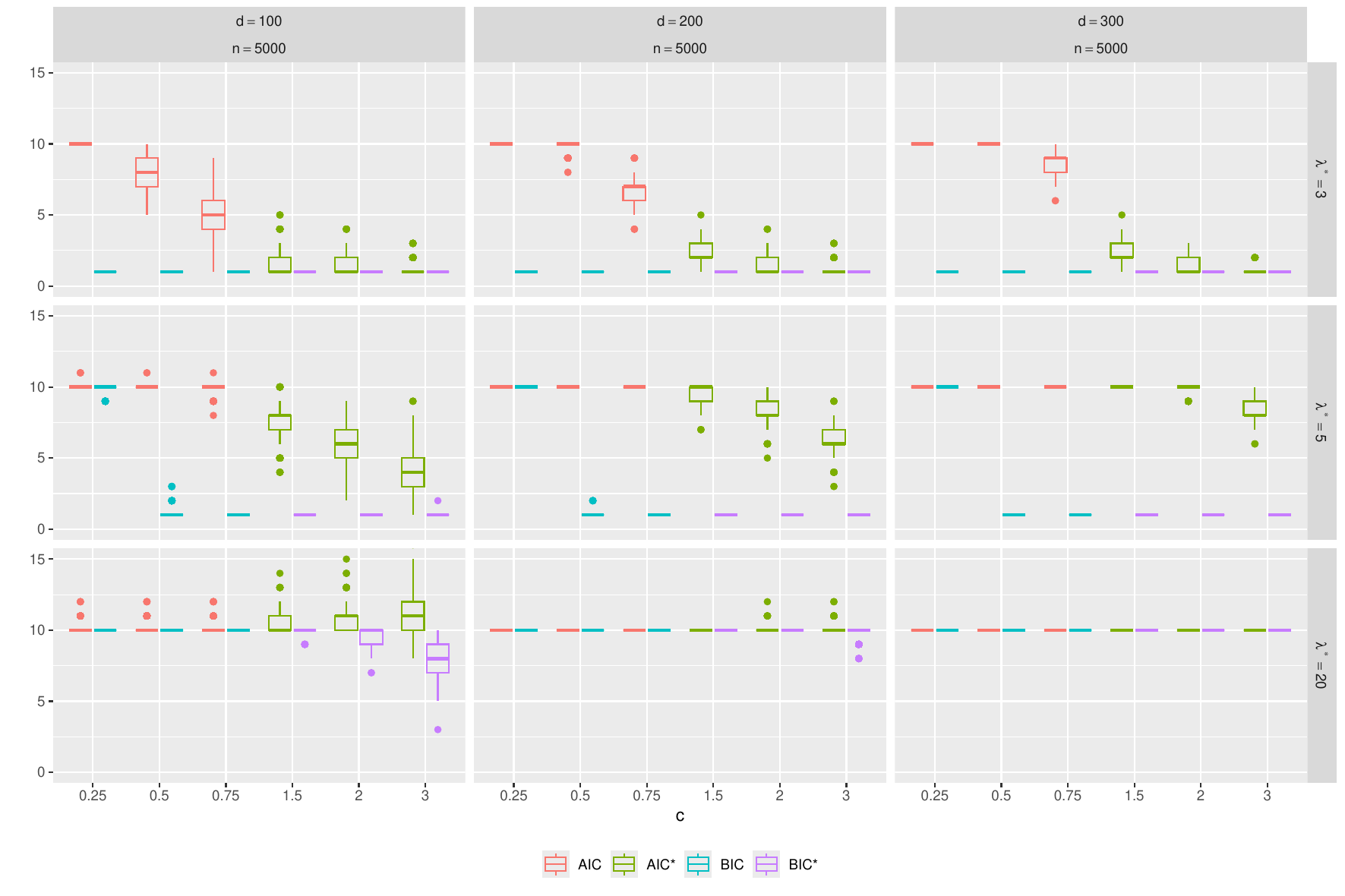}
    \caption{ Simulations for directional factor data with $\ps = 10$ and sample size $n = 5000$: From left to the right the dimension increases from $d=100$, $ d=200$ to $d =300$. From top to bottom, the value of the relevant eigenvalues increases from $ \lambda^* = 3$, $\lambda^* = 5$ to $\lambda^* = 20$. In every subplot the ratio $c = d / k_n$ increases from left to right from  $c= 0.25$, $c = 0.5$, $c = 0.75$, $c = 1.5$, $c = 2$ to $c = 3$.  The box plot shows the estimator $\widehat p_n$ for $\ps=10$ for the different information criteria.}
    \label{fig:DFM_high}
\end{figure}

For the high-dimensional case $d \ge 100$, \Cref{fig:DFM_high} depicts the simulation results. Note that for  $\lambda^* = 3$ the gap condition is satisfied when $c < 1$, and for $\lambda^* = 5$ and $20$  for all $c$. It should also be noted that for fixed $n$ and $d_n$ but increasing $c$, the number of extreme observations $k_n$ decreases, leading to a smaller sample size. The $\AIC$ and $\QAICstar$ both profit from an increase in dimension and $\lambda^*$.   Overall, the estimates of both criteria get better for a larger dimension. In comparison to \Cref{fig:DFM}, we see that the $\QAICstar$ has the tendency to underestimate $\ps$ for $\lambda^* \le 5$, $c = 2$ and $c = 3$, which is consistent with \Cref{th:cons_high_dim_greater_1}. The estimates $\widehat p_n$ of the $\AIC$ and $\QAICstar$ are closer to $\ps$ in comparison to the $\BIC$ and $\QBICstar$ as soon as the gap condition is fulfilled.  When the gap condition is not satisfied, the information criteria underestimate $\ps$, where for $c \ge 0.5$ the $\BIC$ and $\QBICstar$ only give usable results for $\lambda^* = 20$. For $c > 1$ the $\QBICstar$ shows underestimation in all cases. 

\subsection{Directional model with noise} \label{sec:simulation_directional_noise}
In this example, we consider again the directional model with the same choice of distributions as in \Cref{sec:simulation_directional},  but additionally, we add noise. As noise, we use the $d$-dimensional random vector
\begin{align*}
    \bepsilon  \sim \Big\vert \mathcal{N}_{d} \Big( \mathbf{0}_d, \frac{100}{d} \bI_d \Big)\Big\vert,
\end{align*}
where the absolute value is entry-wise. Due to the scaling of the covariance matrix by $100/d$ the variance of the norm of $\bepsilon$ converges as $d \rightarrow \infty$ to $100/\sqrt{2}$ (see \Cref{lemma:norm}).
Then, we construct the regularly varying random vector
\begin{align*}
    \bXn = \frac{ \phantom{\Vert} \sqrtsbSigman \bVn  \phantom{\Vert}}{\Vert \sqrtsbSigman \bVn \Vert} \cdot Z + \bepsilon \in \R^{d_n},
\end{align*}
where $\sbSigman, \bVn$ and $Z$ are defined as in \Cref{sec:simulation_directional} and $\bepsilon$ is given as above.

\begin{figure}[ht]
    \centering
    \includegraphics[width=1\textwidth]{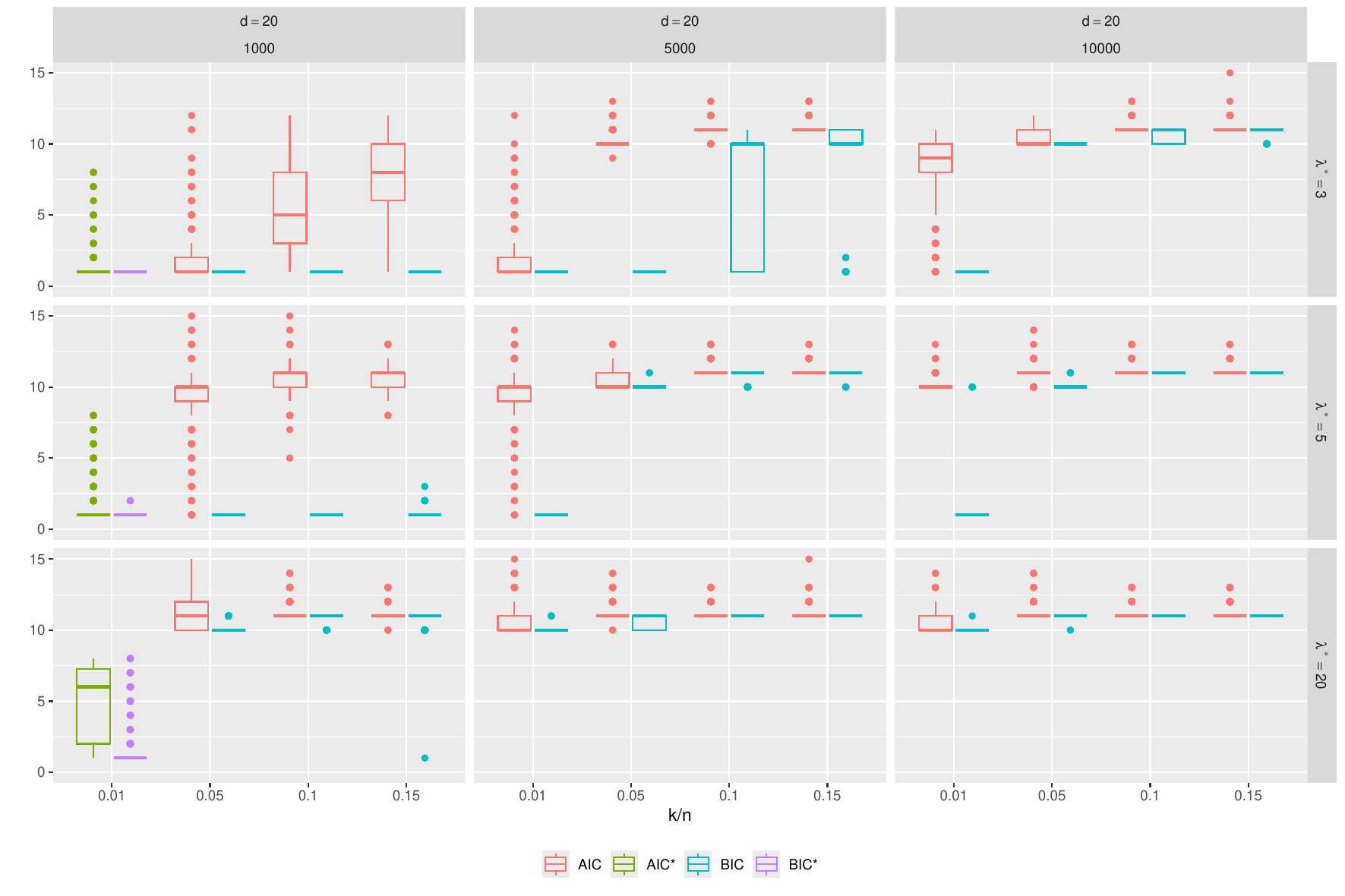}
    \caption{ Simulations for noisy directional factor data with $\ps = 10$ and dimension $d = 20$: From left to the right the  sample size increases from $ n = 1000$, $n = 5000$ to $n = 10000$. From top to bottom, the value of the relevant eigenvalues increases from $ \lambda^* = 3$, $\lambda^* = 5$ to $\lambda^* = 20$. In every subplot the ratio $k_n / n$ increases from left to right from  $0.01$, $0.05$, $0.1$ to $ 0.15$.  The box plots show the estimator $\widehat p_n$ for $\ps=10$ for the $\AIC$ and $\BIC$.}
    \label{fig:DFM_noise}
\end{figure}

The results are shown for $d = 20$  in \Cref{fig:DFM_noise}. Overall, the results are similar to \Cref{fig:DFM}, but with more deviation from the true value $\ps$. In most cases (e.g. $n = 5000, 10000$, $k_n / n \ge 0.05$ and $\lambda^* \ge 5$), the information criteria estimated $\widehat p_n=11$ relevant eigenvalues, therefore identifying not only the 10 dominant eigenvalues but also including the noise. The noise leads to more fluctuation of the $\AIC$ estimates, especially to overestimation of $\ps$. For the $\BIC$ there are cases (e.g. $n = 1000$, $\lambda^* \leq 5$ and $k_n/n = 0.15$), where the $\BIC$ estimates $\widehat p_n=1$ instead $\ps = 10$ and without noise the estimate is concentrated near $\ps=10$. The $\AIC$ does not show this behavior. The influence of the noise decreases for larger $\lambda^*$, resulting in a larger spike.

\begin{figure}[ht]
    \centering
    \includegraphics[width=1\textwidth]{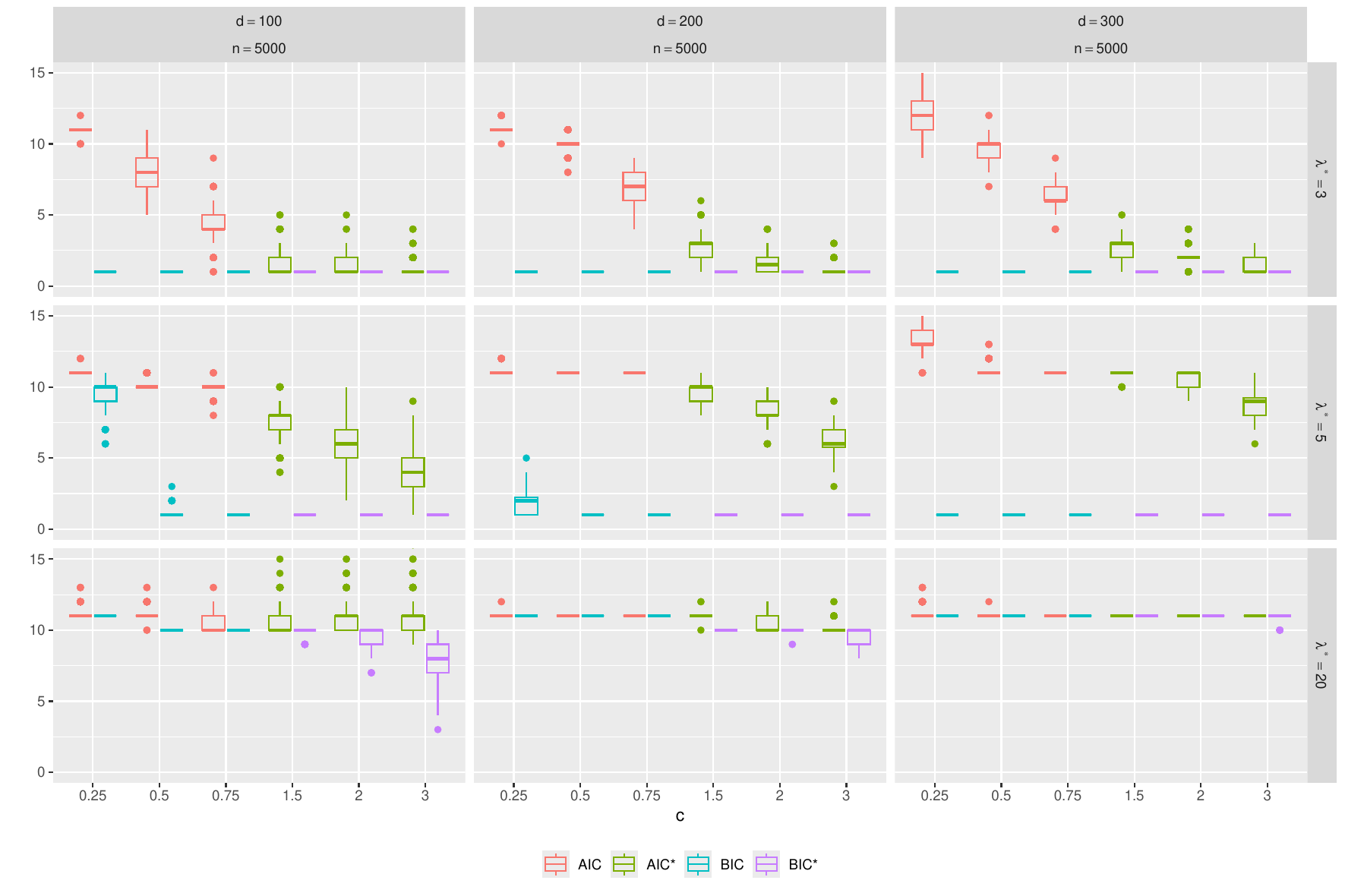}
    \caption{ Simulations for noisy directional factor data with $\ps = 10$ and sample size $n = 5000$: From left to the right the dimension increases from $d=100$, $ d=200$ to $d =300$. From top to bottom, the value of the relevant eigenvalues increases from $ \lambda^* = 3$, $\lambda^* = 5$ to $\lambda^* = 20$. In every subplot the ratio $c = d / k_n$ increases from left to right from  $c= 0.25$, $c = 0.5$, $c = 0.75$, $c = 1.5$, $c = 2$ to $c = 3$.  The box plot shows the estimator $\widehat p_n$ for $\ps=10$ for the different information criteria.}
    \label{fig:DFM_high_noise}
\end{figure}

\Cref{fig:DFM_high_noise}
provides a visualization of the results in the high-dimensional cases $d= 100$, $d= 200$ and $d= 300$. We see that the effect of the noise is similar to the low-dimensional case. The overall fluctuation increases compared to the simulation without noise in \Cref{fig:DFM_high}. The $\AIC$ and $\QAICstar$ estimate the noise as an additional direction, {for example, when $\lambda^* = 20$ and $d = 300$.}  The $\BIC$ and the $\QBICstar$ give as estimation $\widehat p_n=10$ in some cases (e.g. $d = 100$, $\lambda^* = 20$ and $c \ge 0.5$), therefore they are able to differentiate between the noise and the true value $\ps = 10$.

\subsection{Spiked angular Gaussian model} \label{sec:AGM}

 In this section, we consider the contaminated spiked angular Gaussian model, which can also be found in \citet{AMDS:22}. For  $1 \le \ps \le d$
we define the  regularly varying random vector
\begin{align*}
    \bX =  \bN Z \in \R^d,
\end{align*}
where $Z$ is a univariate standard Fréchet distributed random variable, $\bN$ follows a $d$-dimensional centered normal distribution with covariance matrix
\begin{align*}
    \bH \coloneqq \sum_{i=1}^{\ps} \lambda_i \bv_i \bv_i^\top + \lambda \bI_d,
\end{align*}
where $\bv_i, \, i = 1, \ldots, {\ps}$ are orthogonal vectors and  $\lambda_1 \ge \lambda_2 \ge \cdots \ge \lambda_{\ps} > \lambda = \cdots = \lambda > 0$. %and $\bepsilon$ is the entry-wise absolute value of a standard $d$-dimensional normal distribution multiplied with a univariate $\mathrm{Fr\acute{e}chet}(2)$ random variable. The parameter $\sigma$ controls the amount of noise included.
Note that the distribution of $\bX$ differs from the directional model in \Cref{sec:simulation_directional}, since the normal distribution is not standardized when $\bX$ is generated. 
The spectral vector arising from $\bX$ concentrates on a $p$-dimensional subspace and is given by (see \citet{AMDS:22})
\begin{align*}
    \P( \bTheta \in \cdot) = \frac{\E  [  \Vert \bN \Vert \delta_{\bN / \Vert \bN \Vert } ( \cdot ) ] }{ \E[ \Vert \bN \Vert ]}.
\end{align*}
\begin{figure}[ht]
    \centering
    \includegraphics[width=1\textwidth]{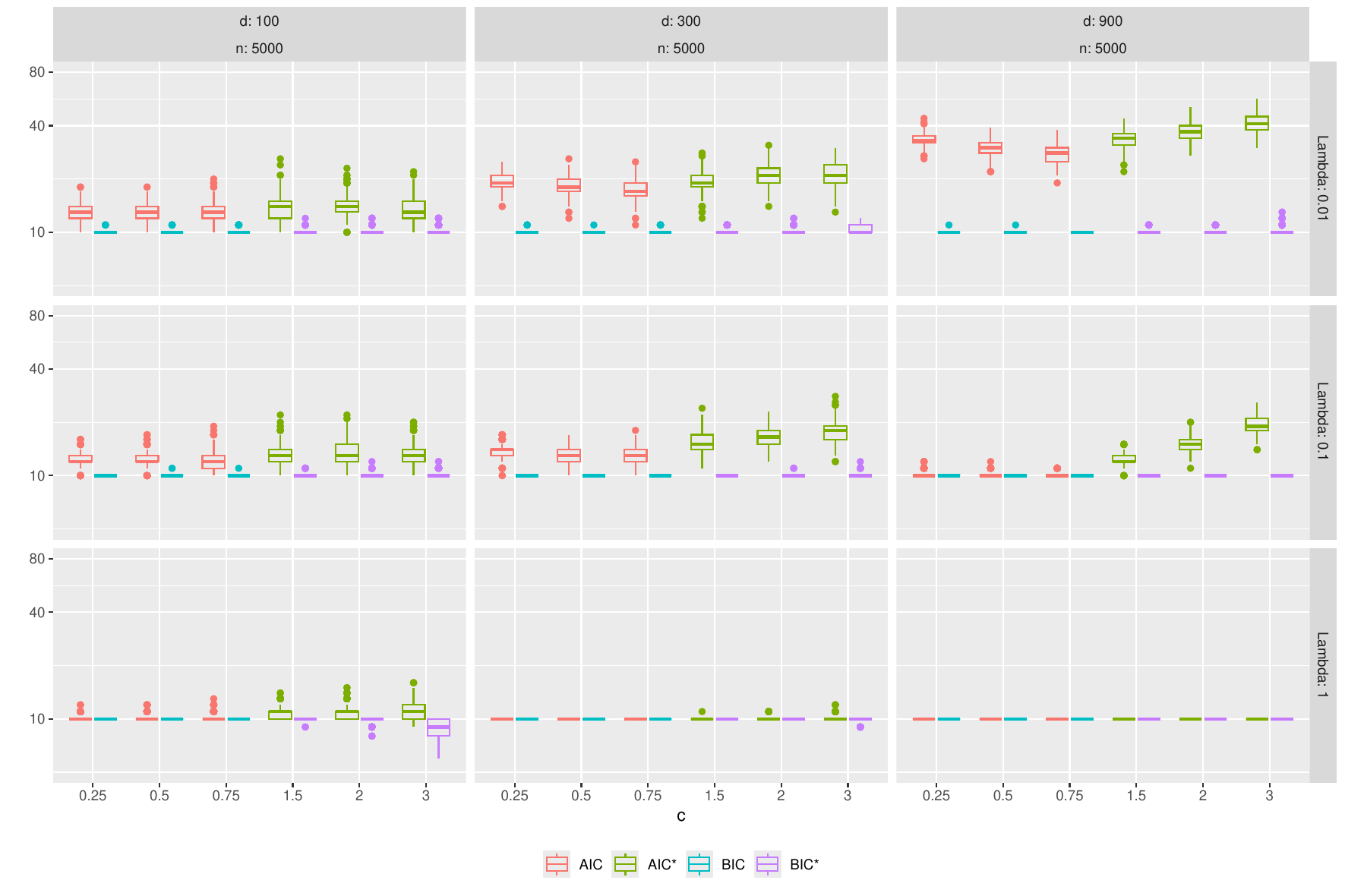}
    \caption{Simulations for spiked angular Gaussian data with $\ps = 10$: From left to the right the dimension increases from $d=100$, $d =300$ to $d =900$. From top to bottom, the value of $\lambda$ increases from $\lambda = 0.01$, $\lambda = 0.1$ to $\lambda = 1$. In every subplot the ratio $c = d / k_n$ increases from left to right from  $c= 0.25$, $c= 0.5$, $c = 0.75$, $c = 1.5$, $c = 2$ to $c = 3$.  The box plot is log-scaled and shows the estimator $\widehat p_n$ for the different information criteria.}
    \label{fig:AGM}
\end{figure}

For the comparison, we run simulations with sample size $n =5000$, dimension $d=100$, $d =300$ to $d =900$.
The matrix $\bH$ is fixed for each sample but is initially randomly generated for the simulation, where the eigenvalues $\lambda_1= \cdots= \lambda_{10}=20$ are equal to $20$, $\ps=10$ and the last eigenvalue $\lambda$
varies; we analyze the behavior of the information criteria when $\lambda$ gets closer to $0$ and thus, the spiked covariance assumption is closer to being violated.
Therefore, we compare the results for $\lambda = 0.01$, $\lambda = 0.1,$ and $\lambda = 1$. The eigenvectors $\bv_i$ are generated with the R package pracma. 

 The results are illustrated in \Cref{fig:AGM}. It is evident that, when the gap is sufficiently large, then the $\BIC$ and $\QBICstar$ are less affected by a small eigenvalue $\lambda$ than the $\AIC$ and $\QAICstar$. The smaller $\lambda$ is chosen, the larger the overestimation of the $\AIC$ and $\QAICstar$ is, whereby for $d = 900$ and $\lambda = 0.01, 0.1$ the $\QAICstar$ overestimates $\ps$ more than $\AIC$. When $\lambda = 1$ and $d \ge 300$ the performance of all criteria is nearly identical.

\section{Application to precipitation data} \label{sec:data_set}
In this section, the information criteria are applied to precipitation data in Germany taken from \citet{DWD}. The data set consists of daily station observations of the precipitation height for Germany between  January 1, 1951 and  March 31, 2022 at $d = 500$ stations. The stations are marked by black dots in \Cref{fig:map_of_germany_stations}.
 The data is preprocessed to include only observations from January, February and March, and transformed to standard Fréchet margins. After data cleaning, the resultant dataset contains $n = 2546$ observations, each with precipitation records from at least one station. In \Cref{fig:map_of_germany_stations} we see the stations of the empirical eigenvectors $\widehat \bv_i=(v_{1}^{(i)},\ldots,v_{d}^{(i)})^\top$, where $v_{j}^{(i)}\geq 0.6 v_{(1)}^{(i)}$, $i=1,\ldots,5$, of the 5 largest empirical eigenvalues $\well_{n,1}, \ldots, \well_{n,5}$ if $k_n=76$; the stations of each eigenvector are colored differently.

\begin{figure}
     \centering
     \begin{subfigure}[t]{0.49\textwidth}
        \centering
         \includegraphics[width=0.7\textwidth]{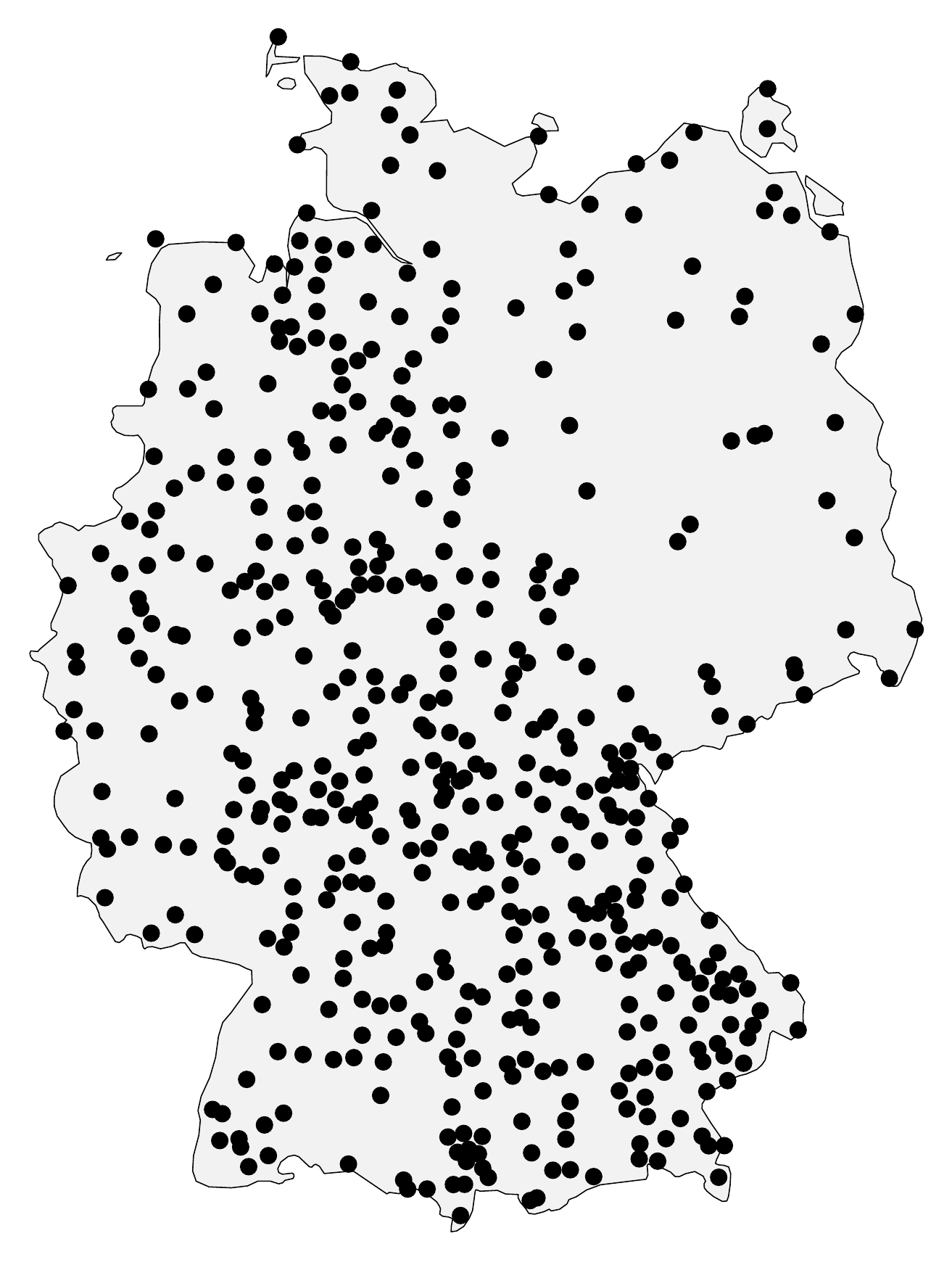}
         %\label{fig:map_of_germany_stations}
     \end{subfigure}
     \hfill
     \begin{subfigure}[t]{0.49\textwidth}
         \centering
         \includegraphics[width=0.7\textwidth]{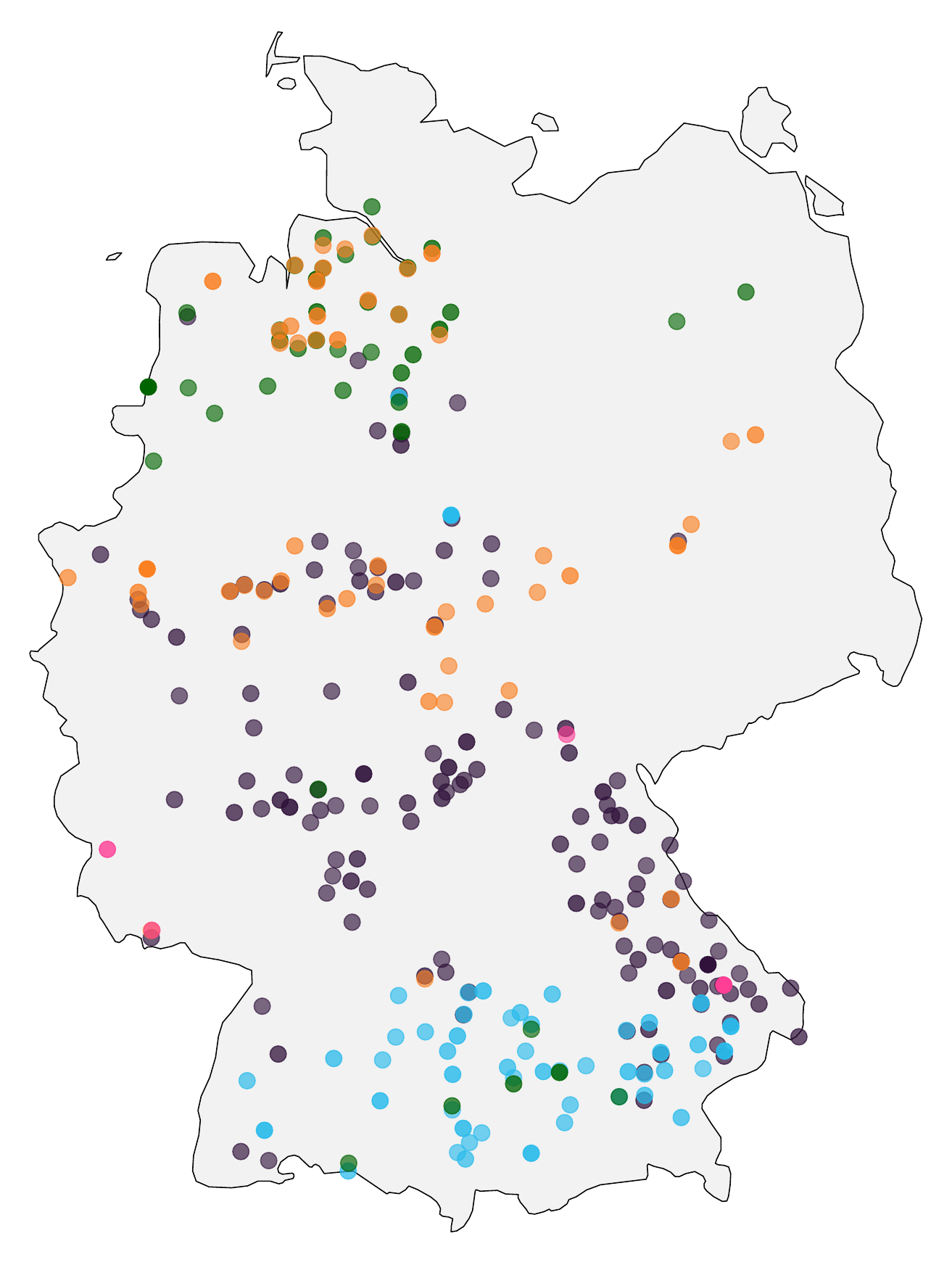}
         %\label{fig:map_of_germany_stations2}
     \end{subfigure}
        \caption{Left figure: Map of Germany with all stations highlighted by black dots. Right figure: Map of Germany with the most extreme stations of the empirical eigenvectors of the five largest empirical eigenvalues, colored by eigenvectors. }
        \label{fig:map_of_germany_stations}
\end{figure}

We consider $1\%$ to $15\%$ of the data as extreme, corresponding to $25$ to $382$ observations. In these cases $d > k_n$ and therefore, we assume to be in the high-dimensional setting with $c > 1$ and apply $\QAICstar$ and
$\QBICstar$ from \Cref{def:aic_high_dim}. The number of candidate models $q_n$ for the $\QAICstar$ is chosen as $d/2 = 250$ to account for the assumption of \Cref{th:cons_high_dim_greater_1}. 

\begin{figure}[ht]
    \centering
    \includegraphics[width=1\textwidth]{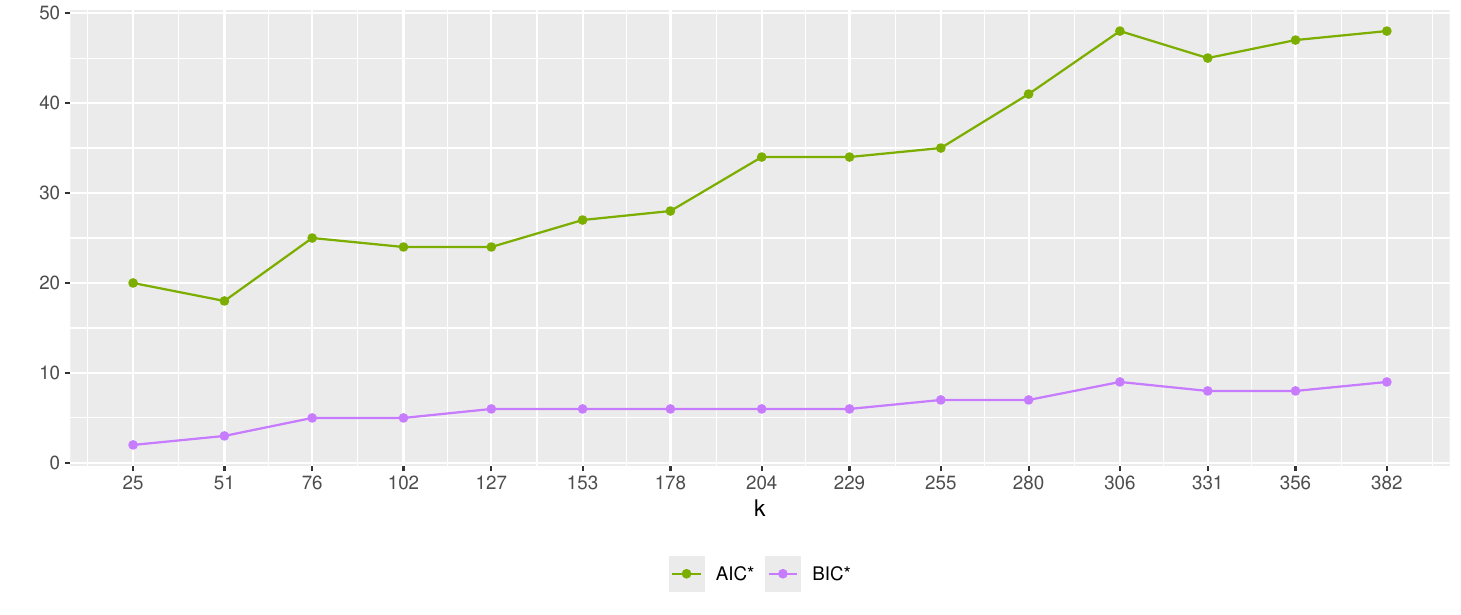}
    \caption{The estimated number $\widehat p_n$ of significant eigenvalues  determined by $\QAICstar$ and $\QBICstar$ plotted against $k_n$. \hspace{10cm}}
    \label{fig:precipitation}
\end{figure}

\begin{figure}[ht]
    \centering
    \includegraphics[width=0.9\textwidth]{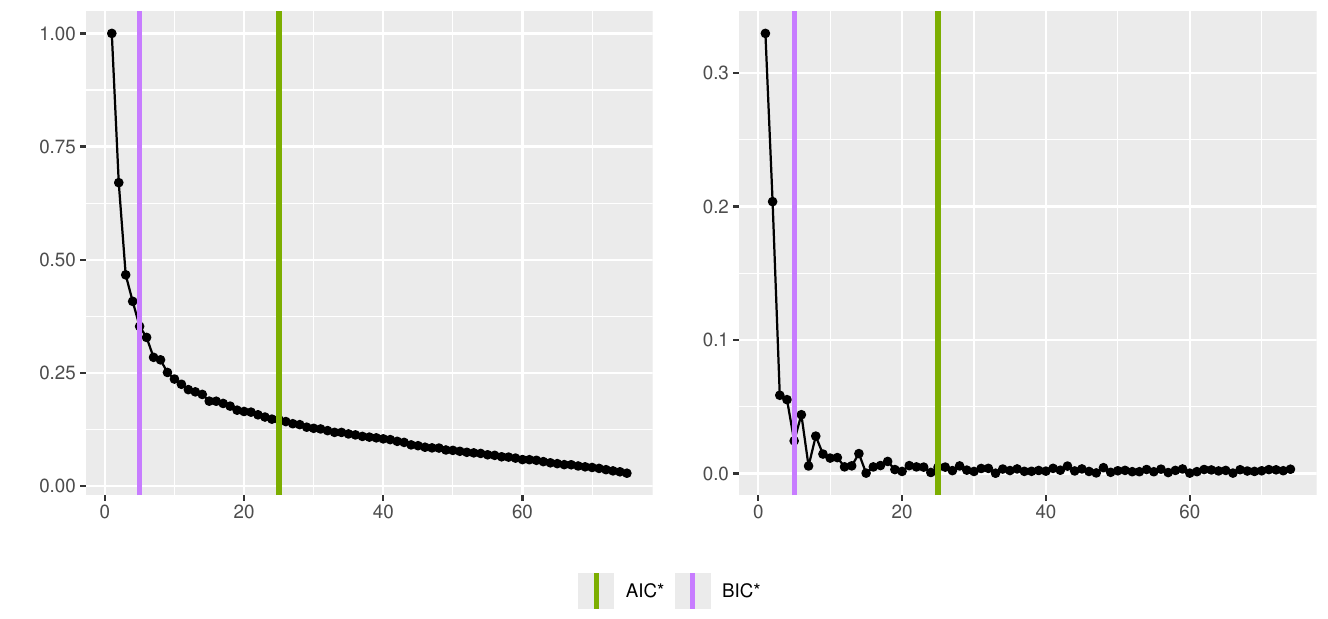}
    \caption{ For $k_n = 76$, on the left hand side the scaled ordered empirical eigenvalues $\well_{n,i}/\well_{n,1}$, $i=1,\ldots,75$  and on the right hand side the differences of the ordered empirical eigenvalues divided by the value of the largest eigenvalue $(\well_{n,i}-\well_{n,i+1})/\well_{n,1}$, $i=1,\ldots,75$ are plotted. The vertical line indicates the  $\QAICstar$ estimator $\widehat p_n=25$ and the $\QBICstar$ estimator $\widehat p_n=5$.}
    \label{fig:precipitation_eigenvalues}
\end{figure}

\begin{table}[ht]
    \centering
    \begin{tabular}[t]{c||c|c|c|c|c|c|c|c|c|c|c|c|c|c|c}
    %\toprule
       $k_n$  & 25 & 51 & 76 & 102 & 127 & 153 & 178 & 204 & 229 & 255 & 280 & 306 & 331 & 356 & 382\\
       \hline%\midrule
        $\QAICstar$ & 20 & 18 & 25 & 24 & 24 & 27 & 28 & 34 & 34 & 35 & 41 & 48 & 45 & 47 & 48\\
        \hline
        $\QBICstar$ & 2 & 3 & 5 & 5 & 6 & 6 & 6 & 6 & 6 & 7 & 7 & 9 & 8 & 8 & 9 \\
        %\bottomrule
    \end{tabular}
    \caption{ Estimated number $\widehat p_n$ of significant eigenvalues in the precipitation data for different values of $k_n$ and information criteria $\QAICstar$  and $\QBICstar$.}
    \label{tab:my_label}
\end{table}

\Cref{fig:precipitation} shows the number of estimated significant eigenvalues $\widehat{p}_n$ mapped against $k_n$. 
 The estimates using $\QAICstar$ stabilize between $k_n = 76$ and $k_n = 178$, ranging from $24$ to $28$, before increasing further. In contrast, $\QBICstar$ stabilizes for $k_n$ between $76$ and $229$, with values of $5$ and $6$. Even for $k_n \geq 255$, the $\QBICstar$ remains between $7$ and $9$, whereas $\QAICstar$ continues to increase. This difference between the estimates aligns with the heavier penalty imposed by $\QBICstar$, which leads to smaller estimates compared to $\QAICstar$.
 These estimates reduce the dimensionality of $d = 500$ by factors of $20$ and $100$, respectively. For comparison of these different estimates, the scaled empirical eigenvalues $\well_{n,i}/\well_{n,1}$, $i=1,\ldots,75$, are plotted in 
 the left picture of \Cref{fig:precipitation_eigenvalues}. At first view, they seem not to be constant after some point, contradicting the spiked covariance assumption. But if we investigate the scaled increments of the eigenvalues $(\well_{n,i}-\well_{n,i+1})/\well_{n,1}$, $i=1,\ldots,75$ in the right plot of \Cref{fig:precipitation_eigenvalues}, we realize that after some point these increments are nearly constant. The information criteria seem to estimate the point where these increments are constant because in the interval $[5,24]$, which is spanned by our estimators, this happens. 
 
\section{Conclusion} \label{sec:conclusion}
The paper proposed information criteria based on the AIC and BIC for Gaussian random vectors to detect the number $p^*$ of significant principal components in multivariate extremes, which corresponds to the location of the spike in the eigenvalues of the covariance matrix of the angular measure.
Our analysis encompassed both the classical large-sample setting and the high-dimensional setting,  which has become increasingly relevant for extreme value theory in today's applications. We established the consistency of the $\BIC$ in the large-sample setting and sufficient criteria for the $\AIC$ and the $\BIC$ to be consistent in the high-dimensional setting of a directional model. The results of this paper are in accordance with the results in the non-extreme world. For the proofs we derived some new results on the asymptotic properties of the empirical eigenvalues of $\bSigma$ in both the large-dimensional case, but, in particular, in the high-dimensional case using methods from random matrix theory. The performance of the information criteria was further validated through a simulation study and a real-world example. 

The case $c=0$ was not covered in that paper because the $\AIC$ and $\BIC$, as defined here, are not consistent even when a type of gap condition is satisfied. From a practical point of view, we believe that in the context of multivariate extremes, this is also not a realistic scenario because usually $k_n$, the number of extremes, is small, and then $d_n/k_n$ will be large.

The paper
focused on eigenvalues of $\bSigma$  that satisfy the spiked covariance structure in \eqref{sec:model}, where 
$\slambdaohnen_{\ps}$ is a distant spiked eigenvalue in the sense that
 $\slambdaohnen_{\ps} > 1 + \sqrt{c}$ and $c=\lim_{n\to\infty}d_n/k_n>0$. 
For applications, these eigenvalue assumptions are restrictive, as we see in our data example in \Cref{fig:precipitation_eigenvalues}, where the empirical eigenvalues decrease but do not stabilize at some point. Therefore, it is worth exploring more general eigenvalue structures of the covariance matrix to estimate the number of significant components of $\bSigma$ such as, e.g., $\slambdaohnen_{\ps} > 1 + \sqrt{c}$ and $\slambdaohnen_{\ps+1}<1 + \sqrt{c}$ where all eigenvalues 
$\slambdaohnen_{j}$ for $j=\ps+1,\ldots,d_n$ are in a  neighborhood of $1$ or $0$.

Additionally, as a starting point of this line of research on PCA for high-dimensional extremes, the consistency results of the information criteria were based on the assumption that the underlying model is a directional model, similar to multivariate statistics, where the first results were derived for Gaussian models with a special covariance structure. Of course, it would also be interesting to explore generalizations or alternatives to the directional model. 

Finally, we would like to point out that changing $k_n$ changes not only the $\AIC$ and  $\BIC$ estimators $\widehat{p}_n$, but also the empirical eigenvalues and hence, the scree plot as in \Cref{fig:precipitation_eigenvalues}. Therefore, the optimal choice of $k_n$ is nontrivial in this context, and some further research, as discussed in
\cite{butsch2024fasen} and \cite{meyer_muscle23} for the choice of $k_n$ is needed.

\appendix

\section{Proofs of Section 2} \label{sec:asymptotic_proof}
\label{sec:high_dim_asymptotic_proof}

\begin{proof}[Proof of \Cref{cor:ev_sqrtn_normality}]~

\begin{enumerate}[(a)]
    \item We use Theorem A.46 in \citet{Bai:Silverstein:2010}, which states that for Hermitian matrices $\bA, \bB \in \R^{d \times d}$ with eigenvalues $\lambda_i(\bA)$ and $\lambda_i(\bB)$, $i = 1, \ldots, d$, the inequality
\begin{align} \label{eq:Weyl}
    \max_{i = 1, \ldots, d} \vert \lambda_i(\bA) - \lambda_i(\bB) \vert \le \Vert \bA - \bB \Vert
\end{align}
holds.
A conclusion from \Cref{th:asymp_theta_cov} is that $\sqrt{k_n} ( \bhSigma_n-\bSigma  ) = O_\P(1)$  and therefore \Cref{eq:Weyl} yields
 \begin{align*}
       (\well_{n,1}, \ldots, \well_{n,d-1}) = (\lambda_1, \ldots, \lambda_{d-1}) + O_\P(1/ \sqrt{k_n})  .  %\label{eq:L_repr}
 \end{align*}

  \item The result corresponds to \citet[Proposition 8]{Dauxois:Pousse:Romain:1982}, which is based on a similar convergence as \Cref{th:asymp_theta_cov}.
\end{enumerate}
\end{proof}

\begin{proof}[Proof of Lemma \ref{lem:hd_spiked_covariance_structure}]
Note that
\begin{align*}
    \bSigman = \Cov( \bThetan)  = {\sbSigman}^{1/2} \Cov \left(  \frac{ \bVn }{\Vert {\sbSigman}^{1/2} \bVn \Vert} \right) {\sbSigman}^{1/2} .
\end{align*}
Utilizing the spectral decomposition $\sbSigman = \bW \bXin {\bW}^\top,$ where $\bW = (\bW_1, \ldots,$ $\bW_{d_n})$ is a $d_n \times d_n$-dimensional orthogonal matrix and 
\begin{equation*}
    \bXin \coloneqq \diag( \obXi, \bI_{d_n - \ps} ) : = \diag( \slambda{1}, \ldots, \slambda{\ps}, 1, \ldots, 1)\in\R^{d_n\times d_n}
\end{equation*}
is a diagonal matrix consisting of the eigenvalues of $\sbSigman$, we receive with $\Vert \bW \bx \Vert = \Vert  \bx \Vert$ for $\bx \in \R^{d_n}$ that
\begin{align*}
        \bSigman=\Cov \left(  \frac{{\sbSigman}^{1/2} \bVn }{\Vert {\sbSigman}^{1/2} \bVn \Vert} \right) &=  \bW \Cov \left(  \frac{ { \bXin}^{1/2} \bVn }{\Vert  { \bXin}^{1/2} \bVn \Vert} \right) {\bW}^\top.
\end{align*}
Hence, the matrices
\begin{align*}
    \Cov \left(  \frac{{\sbSigman}^{1/2} \bVn }{\Vert {\sbSigman}^{1/2} \bVn \Vert} \right) 
\quad
\text{ and } \quad
    \Cov \left(  \frac{{\bXin}^{1/2} \bVn }{\Vert {\bXin}^{1/2} \bVn \Vert} \right) 
\end{align*}
are similar and share the same eigenvalues (\citet[Theorem 1.3.22]{matrix_analysis}). Therefore, we  assume in the following w.l.o.g. that
    $\sbSigman = \bXin$
and hence,
\begin{align*}
    \bSigman  = \Cov \left(  \frac{{\bXin}^{1/2} \bVn }{\Vert {\bXin}^{1/2} \bVn \Vert} \right) 
\end{align*}
is a diagonal matrix. Indeed, since $\bXin$ is a diagonal matrix and $V_1, \ldots, V_{d_n}$ are symmetric and i.i.d., the components of $ {\bXin}^{1/2} \bVn /\Vert {\bXin}^{1/2} \bVn \Vert$ are uncorrelated.
Further, the eigenvalues of $\bSigman$ are the diagonal entries
\begin{align*}
    \diag(\bSigman)_i = \E \left[ \frac{\slambda{i} V_i^2}{\Vert {\sbSigman}^{1/2} \bVn \Vert^2} \right] = \E \left[ \frac{\slambda{i} V_i^2}{ \sum_{j=1}^{\ps}  \slambda{j} V_j^2 + \sum_{j=\ps + 1}^{d_n} V_j^2}  \right], \quad i = 1, \ldots, \ps
\end{align*}
and  
\begin{align*}
    \diag(\bSigman)_i=\diag(\bSigman)_{d_n}=\E \left[ \frac{ V_{d_n}^2}{ \sum_{j=1}^{\ps}  \slambda{j} V_j^2 + \sum_{j=\ps + 1}^{d_n} V_j^2}  \right],
     \quad i = \ps+1,\ldots,d_n,
\end{align*}
which has multiplicity $(d_n - \ps)$.
For $1 \le i \le \ps$  and $l > \ps$, the function
\begin{align*} 
    \frac{\xi V_i^2  - V_l^2}{\xi V_i^2 +  \sum_{ \substack{j=1\\ j \ne i}}^{\ps}  \slambda{j} V_j^2 + \sum_{j=\ps + 1}^{d_n} V_j^2} %\label{eq:monoton_ev}
\end{align*}
is strictly increasing function in $\xi$ since  the derivative in $\xi$ is strictly positive.
A conclusion is then for $1 \le i \le \ps$ with $\slambda{i} > 1$ and $l > \ps$ that 
\begin{align*}
    \diag(\bSigman)_i - \diag(\bSigman)_l &= \E \left[  \frac{\slambda{i} V_i^2 -  V_l^2}{ \sum_{j=1}^{\ps}  \slambda{\ps} V_j^2 + \sum_{j=\ps + 1}^{d_n} V_j^2 }  \right] \\
    &> \E \left[  \frac{ V_i^2 -  V_l^2}{ \sum_{\substack{j=1 \\ j \ne i}}^{\ps}  \slambda{\ps} V_j^2 + V_i^2+ \sum_{j=\ps + 1}^{d_n} V_j^2 }  \right]=0.
\end{align*}
Therefore, we receive that the first $\ps$ diagonal entries of $\bSigman$ correspond to the $\ps$ largest eigenvalues of $\bSigman$ namely $\diag(\bSigman)_1$, \ldots, $\diag(\bSigman)_{p^*}$ and the remaining $(d_n-p^*)$ eigenvalues are strictly smaller and identical to $\diag(\bSigman)_{d_n}$. 
\end{proof}

\subsection{Proof of \Cref{lem:Bai_Yaio2}} %\Cref{lem:Bai_Yaio2}}
For the proof of \Cref{lem:Bai_Yaio2} we combine ideas for the spiked covariance model from \citet{johnstone:2018} and for compositional data from \citet{jiang:2023}.
First, we derive an alternative representation for $\bhSigman$ in \Cref{def:sigma_estimator_hd}. As a consequence of the independence between the radial components $Z_1, \ldots, Z_n$ and the directional components  $\bXn_1/\Vert \bXn_1 \Vert, \ldots,  \bXn_{k_n} /\Vert \bXn_{k_n} \Vert$ we obtain 
\begin{align*}
  \widehat{\bTheta}{}^{(n)}   &= \sum_{j=1}^n \frac{\sqrtsbSigman \bVn_{j}}{\Vert \sqrtsbSigman \bVn_{j} \Vert} \mathbbm{1} \{Z_j   >  Z_{(k_n +1, n)}   \} \overset{\mathcal{D}}{=} \sum_{j=1}^{k_n} \frac{\bXn_j}{\Vert \bXn_j \Vert},
\end{align*}
and similarly
\begin{align}
    \bhSigmanprime \coloneqq{}&  \frac{1}{k_n } \sum_{j=1}^{k_n} \left( \frac{\bXn_j}{\Vert \bXn_j \Vert} - \frac{1}{k_n} \sum_{i=1}^{k_n} \frac{\bXn_i}{\Vert \bXn_i \Vert} \right) \left( \frac{\bXn_j}{\Vert \bXn_j \Vert} - \frac{1}{k_n} \sum_{i=1}^{k_n} \frac{\bXn_i}{\Vert \bXn_i \Vert} \right)^\top  
     \overset{\mathcal{D}}{=}{}   \bhSigman. \label{eq:cov_distr_equality}
\end{align}
The eigenvalues of $\bhSigmanprime$ are denoted by $\well_{n,1}', \ldots, \well_{n,d_n}'$ and due to  \Cref{eq:cov_distr_equality} we receive that
\begin{align} \label{eq lambda}
    (\well_{n,1}', \ldots, \well_{n,d_n}') \overset{\mathcal{D}}{=} (\well_{n,1}, \ldots, \well_{n,d_n}).
\end{align}
Thus, to prove \Cref{lem:Bai_Yaio2} it suffices to derive the asymptotic behavior of $(\well_{n,1}', \ldots, \well_{n,d_n}')$, which relies on the spectral analysis of the empirical covariance matrix of $\sqrtsbSigman \bVn$.  
Therefore, assume that $\bVn_{1}, \ldots, \bVn_{k_n}$ is an i.i.d. sequence with distribution $\bVn$, i.e. $ \bVn_{i} \in \R^{d_n}$ has i.i.d. entries with mean $0$ and variance $1$. Then we define the sequence of matrices 
\begin{align}
    \sbhSigman \coloneqq \frac{1}{k_n} \sum_{i=1}^{k_n} & \bigg( \sqrtsbSigman \bVn_{i} - \frac{1}{k_n} \sum_{j=1}^{k_n} \sqrtsbSigman \bVn_{j} \bigg)  \nonumber \\
    & \quad  \cdot \bigg( \sqrtsbSigman \bVn_{i} - \frac{1}{k_n} \sum_{j=1}^{k_n} \sqrtsbSigman \bVn_{j} \bigg)^\top, \quad n \in \N, \label{eq:def_gamma_empricial_cov}
\end{align}
whose eigenvalues are denoted by $\hslambda_{n,1} > \cdots > \hslambda_{n, d_n} > 0$. 
The aim now is to write $\bhSigmanprime$ and $ \sbhSigman$ as matrix products. Therefore, define 
\begin{align*}
    \wbVn \coloneqq (\bVn_{1}, \ldots, \bVn_{k_n}) \in \R^{d_n \times k_n}
\end{align*}
and 
\begin{align*}
    \bTn \coloneqq \diag( \Vert \sqrtsbSigman \bVn_{1} \Vert^{-1}, \ldots, \Vert \sqrtsbSigman \bVn_{k_n} \Vert^{-1}) \in \R^{k_n \times k_n},
\end{align*}
  which allows us to write 
\begin{align*}
 \biggl( \frac{\bXn_1}{\Vert \bXn_1 \Vert}, \ldots, \frac{\bXn_{k_n}}{\Vert \bXn_{k_n} \Vert} \biggr)^{\top}  
 &= \sqrtsbSigman  \wbVn \bTn .
\end{align*} 
Finally, with the projection matrix $\bCn \coloneqq (\bI_{k_n} -  \mathbf{1}_{k_n} \mathbf{1}_{k_n}^\top /{k_n})$, the  matrices  $ \bhSigmanprime$ and $\sbhSigman$, as defined in \Cref{eq:def_gamma_empricial_cov,eq:cov_distr_equality}, can be written as
\begin{equation}
    \begin{aligned} \label{eq:representation_emp_cov_mat}
     \bhSigmanprime &=  \frac{1}{k_n} ( \wbTheta \bCn ) (  \wbTheta \bCn)^\top, \\ 
     \sbhSigman &= \frac{1}{ k_n} ( \sqrtsbSigman  \wbVn \bCn) ( \sqrtsbSigman  \wbVn \bCn)^\top. 
\end{aligned}
\end{equation}

In the following theorem the connection between the eigenvalues $\hslambda_{n,i}$ and $d_n \wellprime_{n,i}$ is derived.

\begin{theorem} \label{th:ev_high_dim_asymptotic} 
    Let \Cref{asu:high_dim} be given. Suppose that $\bGamman \rightarrow \bGamma$ and $(\slambda{1}, \ldots, \slambda{\ps}) \rightarrow (\slambdaohnen_1, \ldots, \slambdaohnen_{\ps})$ as $\ninf$. If  $\hslambda_{n,1}, \ldots, \hslambda_{n,d_n}$ denote the eigenvalues of $\sbhSigman$ in \Cref{eq:def_gamma_empricial_cov} and $\wellprime_{n,1}, \ldots, \wellprime_{n,d_n}$ denote the eigenvalues of $ \bhSigmanprime$ in \Cref{eq:cov_distr_equality}, then as $n\to\infty$,
    \begin{align*}
        \max_{1 \le i \le d_n}  \big\vert   \hslambda_{n,i}  -  d_n \wellprime_{n,i}  \big\vert \Pconv 0. 
    \end{align*}
\end{theorem}

\begin{proof}[Proof of \Cref{th:ev_high_dim_asymptotic}]
    Due to Theorem A.46 in \citet{Bai:Silverstein:2010} and the sub-multiplicativity of the spectral norm we receive that
     \begin{align}
        \max_{1 \le i \le d_n} \left\vert \sqrt{  \hslambda_{n,i} } - \sqrt{ d_n \wellprime_{n,i} } \right\vert  &\le   \left\Vert \frac{\sqrt{d_n} \sqrtsbSigman  \wbVn  \bTn  \bCn }{\sqrt{k_n}}  - \frac{\sqrtsbSigman  \wbVn \bCn }{ \sqrt{k_n}}  \right\Vert \nonumber \\
        &\le   \left\Vert  \bCn \right\Vert \cdot \left\Vert \sqrt{d_n}   \bTn - \bI_{k_n} \right\Vert \cdot \left\Vert \frac{\sqrtsbSigman  \wbVn}{  \sqrt{k_n}} \right\Vert\nonumber \\
        &=   \left\Vert \sqrt{d_n}   \bTn - \bI_{k_n} \right\Vert \cdot \left\Vert \frac{\sqrtsbSigman  \wbVn}{ \sqrt{k_n}} \right\Vert=:J_n\cdot H_n, \label{eq:ev_high_dim_1}
    \end{align}
where we used that the spectral norm of $\bCn$ is bounded by $1$, because  the only eigenvalues of $\bCn$ are $1$ and $0$ as $\bCn$ is a projection matrix. 

\textbf{Step 1.} First, we show that $J_n$ in \eqref{eq:ev_high_dim_1} converges to $0$
 in probability. Therefore, we use the  partitioning of the random vector $ \sqrtsbSigman \bVn_j$ into the first $\ps$ dependent entries and the remaining $d_n -\ps$ independent entries
\begin{align*}
    \sqrtsbSigman \bVn_j = \begin{pmatrix}
        \bGamman^{1/2} \bVn_{j, \{1, \ldots, \ps\} }\\
        \bVn_{j, \{\ps + 1, \ldots, d_n\} }
    \end{pmatrix}
     \eqqcolon \begin{pmatrix}
        (\Un_{j, 1}, \ldots,\Un_{j,  \ps})^\top\\
        (V_{j , (\ps+1)}, \ldots, V_{j,  d_n})^\top
    \end{pmatrix}. %\label{eq:partition}
\end{align*} 
The eigenvalues of $  (\sqrt{d_n}    \bTn - \bI_{k_n} )$ correspond to the  the diagonal entries.  Since we apply the spectral norm, we receive that
\begin{align*}
    J_n^{1/2}=\left\Vert \sqrt{d_n}   \bTn - \bI_{k_n} \right\Vert^{1/2}=\max_{1 \le i \le k_n}   \left\vert   \frac{  \sqrt{d_n}}{ \big( \sum_{l=1}^{\ps} {\Un_{i,l}}^2 +  \sum_{l={\ps}+1}^{d_n} {V_{i,l}}^2 \big)^{1/2}} -1 \right\vert.
\end{align*}
On the one hand, by  $\E[ {V_{i,j}}^2] = 1,  d_n/k_n \rightarrow c > 0$  and \citet[Lemma 2]{Bai:Yin:1993}  we obtain that as $\ninf$
\begin{align*}
    \max_{1 \le i \le k_n} \left\vert   \frac{  \sum_{l={\ps}+1}^{d_n} {V_{i,l}}^2 }{  d_n} -1 \right\vert  \asconv 0.
\end{align*}
On the other hand,  for $1 \le i \le k_n$,
\begin{align*}
   \sum_{l=1}^{\ps} {\Un_{i,l}}^2 =  \Vert  \bGamman^{1/2} \bVn_{i, \{1, \ldots, \ps\} } \Vert^2 \le \Vert  \bGamman^{1/2} \Vert^2   \Vert  \bVn_{i, \{1, \ldots, \ps\} } \Vert^2 = {\slambda{1}} { \sum_{l=1}^{\ps} {V_{i,l}}^2}. 
\end{align*}
{Since the second moment of $V_{1}^2$ exists, we can conclude from 
Markov inequality  for $\varepsilon >0$ 
\begin{align}
    \P \left( \frac{\slambda{1}}{d_n} \max_{1 \le i \le k_n} \left\vert    \sum_{l=1}^{p^*} {V_{i,l}}^2 \right\vert > \varepsilon\right) &\le \sum_{i=1}^{k_n} \P \left(   \left\vert    \sum_{l=1}^{p^*} {V_{i,l}}^2 \right\vert > \frac{d_n}{\slambda{1}} \varepsilon\right) \nonumber \\
    & = k_n \P \left(   \left\vert    \sum_{l=1}^{p^*} {V_{1,l}}^2 \right\vert > \frac{d_n}{\slambda{1}} \varepsilon\right) \nonumber \\
    &\le k_n \frac{\slambda{1}^2}{d_n^2 \varepsilon^2} \E \left\vert    \sum_{l=1}^{p^*} {V_{1,l}}^2 \right\vert^2 , \label{eq:conv_vn}
\end{align}
where the right-hand side converges to $0$ as $\ninf$, since  $k_n / d_n \to c^{-1}$ and $\slambda{1}^2/d_n \to 0$ as $\ninf$. Therefore, we get
\begin{align*}
     \max_{1 \le i \le k_n}  \left\vert  \frac{  \sum_{l=1}^{\ps} {\Un_{i,l}}^2  }{  d_n} \right\vert &\leq \frac{\slambda{1}}{d_n} \max_{1 \le i \le k_n} \left\vert    \sum_{l=1}^{p^*} {V_{i,l}}^2 \right\vert \Pconv 0.
\end{align*}
}
To summarize,
  \begin{align*}
    \max_{1 \le i \le k_n} &  \left\vert  \left( \frac{ \sum_{l=1}^{\ps} {\Un_{i,l}}^2 +  \sum_{l={\ps}+1}^{d_n} {\Vn_{i,l}}^2 }{  d_n} \right) -1 \right\vert \nonumber \\
    & \le   \max_{1 \le i \le k_n}   \left\vert \left(  \frac{  \sum_{l=1}^{\ps} {\Un_{i,l}}^2  }{  d_n} \right) \right\vert + \max_{1 \le i \le k_n}   \left\vert  \left( \frac{  \sum_{l={\ps}+1}^{d_n} {\Vn_{i,l}}^2 }{  d_n} \right) -1 \right\vert \Pconv 0.
\end{align*}
Finally, by the mean value theorem the inequality 
\begin{align*}
     \left\vert 1 - 1/ \sqrt{x}   \right\vert  
     \le   2 \left\vert x -1 \right\vert \quad \text{ for }  x> \frac{1}{2}
 \end{align*}
holds and hence, as $n\to\infty$,
  \begin{align} 
     J_n^{1/2}=\max_{1 \le i \le k_n}    \left\vert 1 -  \left( \frac{1}{d_n} \sum_{l=1}^{\ps} {\Un_{i,l}}^2 + \frac{1}{d_n} \sum_{l={\ps}+1}^{d_n} {V_{i,l}}^2 \right)^{-1/2} \right\vert \Pconv 0. \label{eq:tn_conv}
 \end{align}
 \textbf{Step 2.} Next, we show that $H_n$ in \eqref{eq:ev_high_dim_1} is $\P$-a.s. bounded.  By    \citet[Theorem 3.1 ]{yin:bai:krishnaiah:1988} (cf. \citet[Theorem 5.8]{Bai:Silverstein:2010}) 
\begin{eqnarray*}
    H_n=\left\Vert \frac{\sqrtsbSigman  \wbVn}{  \sqrt{k_n}} \right\Vert    \le \left\Vert \sqrtsbSigman \right\Vert \cdot \left\Vert \frac{\wbVn}{  \sqrt{k_n}} \right\Vert^2 =  \slambda{1} \frac{\lambda_{\text{max}}({\wbVn}^\top \wbVn)}{k_n} 
    \asconv  \slambdaohnen_{1} 
\end{eqnarray*}
as $n\to\infty$, where $\lambda_{\text{max}}(\cdot)$ denotes the largest eigenvalue of a matrix.

Finally, a combination of \eqref{eq:ev_high_dim_1}, Step 1 and Step 2 result in the statement.
\end{proof}

\begin{remark} \label{rem:order_slambda}
    For the convergence of the right-hand side of \Cref{eq:conv_vn} and hence, \eqref{eq:tn_conv} to zero, it is not necessary that $\slambda{1}$ is bounded; it is sufficient that $\slambda{1} = o(\sqrt{d_n})$ as $n\to\infty$. 
        But if all moments of $V_1$ exist, it is even sufficient to assume that $\slambda{1} = o(d_n^\beta)$ as $n\to\infty$ for some $\beta < 1$. Indeed, 
    we get analog to \Cref{eq:conv_vn} for $\varepsilon > 0$  that
    \begin{align*}
    \P \left( \frac{\slambda{1}}{d_n} \max_{1 \le i \le k_n} \left\vert    \sum_{l=1}^{p^*} {V_{i,l}}^2 \right\vert > \varepsilon\right) & \le k_n \frac{\slambda{1}^{1/(1-\beta)}}{d_n^{1/(1-\beta)}} \varepsilon^{1/(1-\beta)} \E \left[ \left\vert    \sum_{l=1}^{p^*} {V_{1,l}}^2 \right\vert^{1/(1-\beta)}  \right] \to 0
\end{align*}
as $n\to\infty$, since $k_n / d_n \to c$ and $\slambda{1}^{(1-\beta)^{-1}} /d_n^{(1-\beta)^{-1}-1}=(\slambda{1} /d_n^{\beta})^{1/(1-\beta)} = o( 1)$  as $n\to\infty$.
\end{remark}

Next, we repeat results on the asymptotic distribution of the eigenvalues of $\sbhSigman$ which is mainly based on \citet{Bai:Yao:2012} and \citet{Bai:Choi:Fujikoshi:2018}.
 
\begin{lemma} \label{lem:Bai_Yaio} 
    Let   \Cref{asu:high_dim} be given. Suppose that  $\bGamman \rightarrow \bGamma$ and $(\slambda{1}, \ldots, \slambda{\ps}) \rightarrow (\slambdaohnen_1, \ldots, \slambdaohnen_{\ps})$ as $\ninf$ with $\slambdaohnen_{\ps} > 1 + \sqrt{c}$. Then the following statements hold.
    \begin{enumerate}[(a)]
        \item If $1 \le i \le \ps$ (i.e. $\slambdaohnen_i > 1 + \sqrt{c}$), then $\hslambda_{n,i} \asconv \varphi_c(\slambdaohnen_i)$ as $\ninf$.
        \item Define $l^*:=0$ if $ c\leq 1$ and $l^*:=1-c^{-1}$ if $ c> 1$. Then 
        \begin{align*}
     \lim_{n\to\infty}\sup_{\alpha\in(l^*,1)} \left\vert F^{\sbhSigman\,\leftarrow}(\alpha)-F_{c}^{\leftarrow}(\alpha) \right\vert=0 \quad \P\text{-a.s.},
\end{align*}
 where $F^{\sbhSigman\,\leftarrow}$ is the generalized inverse of the empirical spectral distribution function of $\sbhSigman$ and  $F_{c}(x)$ is defined as in \Cref{lem:Bai_Yaio2}.
        \item If  $i_n(\alpha) > \ps$ (i.e. $\slambdaohnen_{i_n(\alpha)} = 1$) and $i_n(\alpha)/d_n \rightarrow \alpha \in (0,1)$ as $\ninf$, then 
        \begin{align*}
       \sup_{\alpha\in(0,1)} \left\vert
        \hslambda_{n,i_n(\alpha)} - F^\leftarrow_c(1 - \alpha) \right\vert&\asconv 0, \quad  \text{ as } \ninf.
        \end{align*}
        In particular, if $(q_n)_{n\in\N}$ is a sequence in $\N$ with $q_n=o(d_n)$ as $n\to\infty$ and $q_n>p^*$, then $\hslambda_{n,q_n}\asconv (1+\sqrt{c})^2$.
        \item Suppose $0 < c \le 1$ and $(q_n)_{n\in\N}$ is a sequence in $\N$ with $q_n=o(d_n)$ as $n\to\infty$.
        Then as $n\to\infty$,
        \begin{eqnarray*}
             \frac{1}{d_n-q_n} \sum_{i=q_n+1}^{d_n} \hslambda_{n,i}  \asconv 1
\end{eqnarray*}
        and for $q_n>p^*$ we receive that $\hslambda_{n,q_n}\asconv (1+\sqrt{c})^2$.
        \item Suppose $c>1$ and $(q_n)_{n\in\N}$ is a sequence in $\N$ with $q_n=o(d_n)$ as $n\to\infty$.
        Then as $n\to\infty$,
        \begin{eqnarray*}
             \frac{1}{k_n-q_n} \sum_{i=q_n+1}^{k_n} \hslambda_{n,i}  \asconv c
\end{eqnarray*}
        and for $q_n>p^*$ we receive that $\hslambda_{n,q_n}\asconv (1+\sqrt{c})^2$.
    \end{enumerate}
\end{lemma} 
\begin{proof} $\mbox{}$
\begin{enumerate}[(a)]
    \item When the eigenvalues  $(\slambda{1}, \ldots, \slambda{\ps}) = (\slambdaohnen_1, \ldots, \slambdaohnen_{\ps})$ do not depend on $n$, (a) goes back to \citet[Theorem 4.1]{Bai:Yao:2012} (cf.  \citet[Lemma 2.1]{Bai:Choi:Fujikoshi:2018}).
    In the case $\bGamman \rightarrow \bGamma$ and $
        (\slambda{1}, \ldots, \slambda{\ps}) \rightarrow  (\slambdaohnen_1, \ldots, \slambdaohnen_{\ps}) $
    as $\ninf$ the assertion also holds because by \citet[Theorem A.46]{Bai:Silverstein:2010} it can be shown with the same arguments as before that 
    \begin{align*}
        \max_{1 \le i \le \ps} & \left\vert \sqrt{\hslambda_{n,i} ( \bGamman )} - \sqrt{\hslambda_{n,i}( \bGamma)} \right\vert \le \Vert \bGamman - \bGamma  \Vert \left\Vert \frac{\wbVn}{  \sqrt{k_n}} \right\Vert \Vert \bCn \Vert \asconv 0,
    \end{align*}
    where $\hslambda_{n,i} ( \bGamman  )$ and $\hslambda_{n,i} ( \bGamma  )$ is the empirical eigenvalue when $\bGamman$ and $\bGamma$, respectively is used.
    \item The second part is similar to \citet[Theorem 4.1]{Bai:Yao:2012} however, the wording is not clear and therefore we prefer to include the proper statement and proof here.
    Note, if $d_n / k_n \rightarrow c > 0$ as $\ninf$, then for almost all $\omega\in\Omega$, $F^{\sbhSigman}(\omega)$ converges in distribution to $F_{c}$ (cf. \citet[p. 1054]{Bai:Choi:Fujikoshi:2018}, \citet[Theorem 1.1]{silverstein:1995}). This means that there exists a set $\Omega_0\in\mathcal{F}$ with $\P(\Omega_0)=1$ and for any $\omega \in\Omega_0$ and any continuity point $x \in \R$ of $F_c$,
\begin{align*}
    \lim_{n\to\infty}F^{\sbhSigman}(x,\omega)= F_{c}(x). %\label{eq:mp_esd_conv}
\end{align*}
Since the distribution function $F_c$ is continuous on the interval $I:=((1-\sqrt{c})^2,(1+\sqrt{c})^2)$, a conclusion of Polya's Theorem  is the uniform convergence 
\begin{align*}
    \lim_{n\to\infty}\sup_{x\in I} \left\vert F^{\sbhSigman}(x,\omega)-F_{c}(x) \right\vert=0, 
\end{align*}
which implies by \citet[Lemma 1.1.1]{HF:2006} and again Polya's Theorem as well as the uniform convergence of the quantile function
\begin{align*}
     \lim_{n\to\infty}\sup_{\alpha\in(l^*,1)} \left\vert F^{\sbhSigman\,\leftarrow}(\alpha,\omega)-F_{c}^{\leftarrow}(\alpha) \right\vert=0. 
\end{align*}
\item Since $\hslambda_{n,i_n(\alpha)} =F^{\sbhSigman\,\leftarrow}(1- i_n(\alpha)/d_n)$ the statement follows directly from (b).
\item Due to (b), we receive that
\begin{eqnarray*}
             \frac{1}{d_n-q_n} \sum_{i=q_n+1}^{d_n} \hslambda_{n,i} &=&\frac{d_n}{d_n-q_n}\int_0^{1-\frac{q_n}{d_n}} F^{\sbhSigman\,\leftarrow}(1-\alpha)\,d\alpha \asconv  1\cdot \! \int_0^1 \! F_{c}^{\leftarrow}(1-\alpha)\,d\alpha=1.
\end{eqnarray*}
\item Similarly to (d) we have
\begin{eqnarray*}
             \frac{1}{k_n-q_n} \sum_{i=q_n+1}^{k_n} \hslambda_{n,i} &=&\frac{d_n}{k_n-q_n}\int_{1-\frac{k_n}{d_n}}^{1-\frac{q_n}{d_n}} F^{\sbhSigman\,\leftarrow}(1-\alpha)\,d\alpha  \\
             &\asconv &  c\cdot \int_{1-c^{-1}}^1F_{c}^{\leftarrow}(1-\alpha)\,d\alpha=c. 
\end{eqnarray*}
\end{enumerate}
\end{proof}

Finally, we have all auxiliary results for the proof of \Cref{lem:Bai_Yaio2}.

\begin{proof}[Proof of \Cref{lem:Bai_Yaio2}]  (a) An assumption is that $\slambdaohnen_i > 1 + \sqrt{c}$  and hence,  $\slambdaohnen_i$ is a distant spiked eigenvalue for $i = 1, \ldots, \ps$. A conclusion of \Cref{lem:Bai_Yaio}(a) is then that $\hslambda_{n,i} \asconv \varphi_c(\slambdaohnen_i)$. Combined with \Cref{th:ev_high_dim_asymptotic} we receive that
    $d_n \wellprime_{n,i} \Pconv \varphi_c(\slambdaohnen_i)$ as $\ninf$.
    Due to \eqref{eq lambda}, the identical distribution of $\wellprime_{n,i}$ and $\well_{n,i}$, we obtain the final statement, $ d_n \well_{n,i} \Pconv \varphi_c(\slambdaohnen_i)$ as $\ninf$.
    
 Similarly as in (a), the statements  (b)-(d) are combinations of \Cref{lem:Bai_Yaio},
    \Cref{th:ev_high_dim_asymptotic} and \eqref{eq lambda}. 
\end{proof}

\subsection{Proof of \Cref{th:Bai_Lemma2_2_lambda}} %\Cref{th:Bai_Lemma2_2_lambda}}

For the proof of \Cref{th:Bai_Lemma2_2_lambda}, \Cref{th:ev_high_dim_asymptotic} is not useful and an adapted version does not exist. Therefore, the approach is slightly different.
First, the next lemma gives the asymptotic distribution of the eigenvalues from $\bhSigmanprime$, which is then used for the proof of \Cref{th:Bai_Lemma2_2_lambda}. 

\begin{lemma} \label{hilf:Theorem36}
       Let \Cref{asu:high_dim}  with  $\slambda{\ps}  \rightarrow \infty$ and $\slambda{1} = o(d_n^{1/2})$ as $\ninf$ be given. If $i \in\{ 1, \ldots, \ps\}$ then
        \begin{align*}
            \frac{d_n \wellprime_{n,i}}{\slambda{i}} \Pconv 1  \quad \text{ as }  \ninf. 
        \end{align*}
\end{lemma}
\begin{proof}
    We proceed similarly to the proof of \citet[Lemma 2.2]{Bai:Choi:Fujikoshi:2018} and use the spectral decomposition of $\sbSigman$. Let $\sbSigman = \bW \bXin {\bW}^\top,$ where $\bW = (\bW_1, \ldots, \bW_{d_n})$ is a $(d_n \times d_n)$-dimensional orthogonal matrix and $\bXin := \diag( \obXi, \bI_{d_n - \ps} ) : = \diag( \slambda{1}, \ldots, \slambda{\ps}, 1, \ldots, 1)\in\R^{d_n\times d_n}$ consists of the eigenvalues of $\sbSigman$.  Then with  representation \eqref{eq:representation_emp_cov_mat} and
\begin{equation}  \label{def:bAn}
    \bAn \coloneqq  \wbVn \bTn \bCn  \bTn  {\wbVn}^\top
\end{equation} 
we receive 
     \begin{align} \label{sigmastrich1}
         \bhSigmanprime &= \frac{1}{ k_n}  \bW \bXin^{1/2} {\bW}^\top   \wbVn \bTn \bCn  \bTn  {\wbVn}^\top \bW \bXin^{1/2} {\bW}^\top \nonumber\\
         &= \frac{1}{ k_n}  \bW \bXin^{1/2} {\bW}^\top   \bAn  \bW \bXin^{1/2} {\bW}^\top.
     \end{align}
     Further, the eigenvectors are partitioned into the first $\ps$ and the remaining eigenvectors by defining $\obWn =  ( \bW_{1}, \ldots, \bW_{\ps})$ in $\R^{d_n\times p^*}$and  $\tbWn = ( \bW_{\ps+1}, \ldots, \bW_{ d_n})$ in  $\R^{d_n\times (d_n-p^*)}$ such that 
        \begin{align*} %\label{sigmastrich}
                 \bhSigmanprime  &= \frac{ 1}{k_n} \bW \begin{pmatrix}
                     \obXi^{1/2} {\obWn}^\top \bAn  \obWn  \obXi^{1/2} &\, \obXi^{1/2} {\obWn}^\top \bAn  \tbWn   \\                    {\tbWnT}  \bAn  \obWn 
\obXi^{1/2}  &  {\tbWnT}  \bAn \tbWn 
                 \end{pmatrix} {\bW}^\top.
        \end{align*}
        Similarly, we receive with \eqref{eq:representation_emp_cov_mat} and
        \begin{equation} \label{def:bBn}
            \bBn \coloneqq \wbVn  \bCn   {\wbVn}^\top
        \end{equation}         
         that
        \begin{align*}
                 \sbhSigman &= \frac{1}{k_n} \bW  \begin{pmatrix}
                     \obXi^{1/2} \obWn^\top \bBn  \obWn  \obXi^{1/2} & \obXi^{1/2} \obWn^\top \bBn  \tbWn   \\
                    \tbWnT  \bBn  \obWn \obXi^{1/2}  &  \tbWnT  \bBn \tbWn
                 \end{pmatrix} {\bW}^\top.
        \end{align*}
Let $i \in\{1,\ldots,p^*\} $.    
        The Courant-Fischer min-max theorem (\citet[Theorem 4.2.6]{matrix_analysis}) gives
        \begin{align} \label{A.12}
            \frac{d_n \wellprime_{n,i}}{\slambda{i}} = \frac{d_n }{\slambda{i}} \inf_{\bv_1, \ldots, \bv_{i-1} \in \R^{d_n}}  \sup_{\bw \perp \bv_1, \ldots, \bv_{i-1}, \Vert \bw \Vert = 1} \bw^\top \bhSigmanprime \bw .
        \end{align}
    The proof is split into two parts, wherein we establish that ${d_n \wellprime_{n,i}}/{\slambda{i}}$ is bounded below and above by a random variable which converges in probability to $1$ as $\ninf$.  \\
    
    \textbf{Step 1:} First, we derive a lower bound of \eqref{A.12} which converges in probability to $1$.  Therefore, note for arbitrary $\bu_j \in \R^{\dnsub}$ with $ \Vert \bu_j \Vert = 1$ for $1 \le j \le \ps$,  \citet[Lemma A.2 ]{Bai:Choi:Fujikoshi:2018} yields  that as $\ninf$, 
    \begin{align} \label{conv:hilf5}
        \max_{1 \le j \le \ps} \left\vert \bu_j^\top \frac{\bBnsub}{ k_{n}} \bu_j - 1 \right\vert \asconv 0,
    \end{align}
    where $\bBnsub$ is defined as in \Cref{def:bBn}. Now, let $\bAnsub$ be defined as in \Cref{def:bAn}. 
    Then
\begin{align*}
    \left\vert \bu_j^\top \left( \frac{\bBnsub}{ k_{n}} - \frac{\dnsub \bAnsub}{ k_{n}} \right) \bu_j \right\vert &\le \left\Vert  \frac{\bBnsub}{ k_{n}} - \frac{\dnsub \bAnsub}{ k_{n}} \right\Vert \\
    &= \frac{1}{k_{n}}    \left\Vert \wbVnsub \left( \bCnsub - \dnsub \bTnsub \bCnsub {\bTnsub}^\top \right)  {\wbVnsub}^\top \right\Vert\\
    %&\le \frac{1}{k_{n}}   \Vert \wbVnsub  \Vert^2 \left\Vert   \left( \bCnsub - \dnsub \bTnsub \bCnsub + \dnsub \bTnsub \bCnsub - \dnsub \bTnsub \bCnsub {\bTnsub}^\top \right)   \right\Vert\\
    &\le \frac{1}{k_{n}} \Vert \wbVnsub  \Vert^2  \Vert \bCnsub \Vert \left( 1 + \Vert {\sqrt{\dnsub}} \bTnsub \Vert \right) 
     \left\Vert \sqrt{\dnsub}   \bTnsub - \bI_{k_{n}} \right\Vert.  %\asconv 0. 
\end{align*}  
On the one hand,      \citet[Theorem 3.1]{yin:bai:krishnaiah:1988}
implies that $$\frac{1}{k_{n}} \Vert \wbVnsub  \Vert^2 = (\frac{1}{\sqrt{k_{n}}} \Vert \wbVnsub  \Vert)^2 \asconv (1+\sqrt{c})^2.$$ 
On the other hand, 
  since  $\slambdasub{1} = o(\dnsub^{1/2})$ as $n\to\infty$,  a conclusion of  \Cref{rem:order_slambda}  is that $\Vert \sqrt{\dnsub}   \bTnsub - \bI_{k_{n}} \Vert \Pconv  0$ and $\Vert \sqrt{\dnsub}   \bTnsub   \Vert \le \Vert \bI_{k_{n}} \Vert +   \Vert \sqrt{\dnsub}   \bTnsub - \bI_{k_{n}} \Vert\Pconv 1$. In summary, as $\ninf$
\begin{eqnarray} \label{A.17}
    \left\vert \bu_j^\top \left( \frac{\bBnsub}{ k_{n}} - \frac{\dnsub \bAnsub}{ k_{n}} \right) \bu_j \right\vert \le \left\Vert  \frac{\bBnsub}{ k_{n}} - \frac{\dnsub \bAnsub}{ k_{n}} \right\Vert \Pconv 0,
\end{eqnarray}
and finally, using \Cref{conv:hilf5} we have as well 
    \begin{align} \label{conv:hilf6}
        \max_{1 \le j \le \ps} \left\vert \bu_j^\top \frac{\dnsub \bAnsub}{ k_{n}} \bu_j - 1 \right\vert \Pconv 0. 
    \end{align}
Further, for arbitrary vectors $\bv_1, \ldots, \bv_{i-1} \in \R^{\dnsub}$ we take a vector  $\bw_{\bv} = \sum_{j=1}^i a_j \bWsub_j$    orthogonal to  $\bv_1, \ldots, \bv_{i-1}$ with  $ \sum_{j=1}^i a_j^2 = 1$ and hence, $\Vert \bw_{\bv} \Vert = 1$. Since $\bWsub$ is an orthogonal matrix,  we receive with representation \eqref{sigmastrich1} that 
\begin{align*}
     \frac{\dnsub }{\slambdasub{i}} & \bw^\top_{\bv} \bhSigmanprimesub \bw_{\bv} \\
     &=  \frac{\dnsub }{\slambdasub{i}} \sum_{j,l=1}^i a_j a_l {\bWsub_j}^\top    \bWsub \bXinsub^{1/2} {\bWsub}^\top  \frac{\bAnsub}{ k_{n}}    \bWsub \bXinsub^{1/2} {\bWsub}^\top   \bWsub_l \\
     &=   \sum_{j=1}^i a_j^2 \frac{ \slambdasub{j}}{\slambdasub{i}}    {\bWsub_j}^\top   \frac{\dnsub\bAnsub}{k_{n}}     \bWsub_j. 
\end{align*}
A conclusion of  \Cref{A.12},  $\Vert  \bWsub_j\Vert=1$ and \Cref{conv:hilf6}  is then
     \begin{align*}
            \frac{\dnsub \wellprime_{n,i}}{\slambdasub{i}} % &=   \frac{\dnsub }{\slambdasub{i}} \inf_{\bv_1, \ldots, \bv_{i-1} \in \R^{\dnsub}} \sup_{\bw \perp \bv_1, \ldots, \bv_{i-1}, \Vert \bw \Vert = 1} \bw^\top \bhSigmanprimesub \bw \\
             &\ge  \inf_{\bv_1, \ldots, \bv_{i-1} \in \R^{\dnsub}}   \frac{\dnsub }{\slambdasub{i}} \bw_{\bv}^\top \bhSigmanprimesub \bw_{\bv} \\
              &\geq    \inf_{\ba\in\R^{i}:\sum_{j=1}^i a_j^2 = 1} \sum_{j=1}^i a_j^2 %\frac{ \slambdasub{i}}{\slambdasub{i}}  
      {\bWsub_j}^\top   \frac{\dnsub\bAnsub}{k_{n}}     \bWsub_j  \\
             &\ge   1-\max_{1 \le j \le \ps} \left| {\bWsub_j}^\top \frac{\dnsub \bAnsub}{ k_{n}} \bWsub_j - 1 \right| \Pconv 1
         \end{align*}   
as $n\to\infty$.\\
\textbf{Step 2:} Next, we derive an upper bound for \eqref{A.12} which converges in probability to $1$. Therefore, note that
\begin{align*}
     \frac{\dnsub \wellprime_{n,i}}{\slambdasub{i}} &= \frac{\dnsub }{\slambdasub{i}} \inf_{\bv_1, \ldots, \bv_{i-1} \in \R^{\dnsub}}  \sup_{\bw \perp \bv_1, \ldots, \bv_{i-1}, \Vert \bw \Vert = 1} \bw^\top \bhSigmanprimesub \bw \\ 
     &\le \frac{\dnsub }{\slambdasub{i}}   \sup_{\bw \perp \bWsub_1, \ldots, \bWsub_{i-1}, \Vert \bw \Vert = 1} \bw^\top \bhSigmanprimesub \bw.
\end{align*}
Since  $\bWsub_l \perp \bWsub_j$ for $l \neq j$ we can write a vector $\bw \perp \bWsub_1, \ldots, \bWsub_{i-1}$ with $ \Vert \bw \Vert = 1$ as
\begin{align*}
    \bw =  c^2 \bu + (1-c^2) \bv,
\end{align*}
where $c \in [0,1], \bu = \sum_{j=i}^{\ps} a_j \bWsub_j = \obWnsub \ba, \Vert \ba \Vert \! = \!\sum_{j=i}^{\ps} a_j^2 = 1$ and $\bv \! = \!\sum_{j=\ps+1}^{\dnsub} b_j \bWsub_j$ $= \tbWnsub\bb $ satisfying $ \sum_{j=p^*+1}^{\dnsub} b_j^2 = 1$. Recall that $\tbWnTsub \bhSigmanprimesub \tbWnsub= \tbWnTsub  \frac{\bAnsub}{k_{n}} \tbWnsub$. Then,
\begin{align*}
    \frac{\dnsub }{\slambdasub{i}}&   \sup_{\bw \perp \bWsub_1, \ldots, \bWsub_{i-1}, \Vert \bw \Vert = 1} \bw^\top \bhSigmanprimesub \bw \\
    &\le \frac{\dnsub }{\slambdasub{i}} \sup_{c \in [0,1]} \bigg\{ c^2  \sup_{ \substack{ \ba \in \R^{\ps - i +1} ,\\ \Vert \ba \Vert = 1}} \ba^\top \obWnsub^\top \bhSigmanprimesub \obWnsub \ba \\
    &\hspace*{5cm} + (1-c^2) \sup_{\substack{\bb \in \R^{d - \ps}, \\ \Vert \bb \Vert = 1}} \bb^\top \tbWnTsub \bhSigmanprimesub \tbWnsub\bb  \bigg\}\\
   % &= \frac{\dnsub }{\slambdasub{i}} \sup_{c \in [0,1]} \bigg\{ c^2  \sup_{ \substack{ \ba \in \R^{\ps - i} ,\\ \Vert \ba \Vert = 1}} \sum_{j=i}^{\ps} a_j^2 \slambdasub{j} {\bWsub_j}^\top \frac{\bAnsub}{k_{n}} \bWsub_j \\
   % & \hspace{5cm} + (1-c^2) \Vert  \tbWnTsub   \frac{\bAnsub}{k_{n}} \tbWnsub \Vert \bigg\}\\
   % &\le \frac{\dnsub }{\slambdasub{i}} \sup_{c \in [0,1]} \bigg\{ c^2  \sup_{ \substack{ \ba \in \R^{\ps - i} ,\\ \Vert \ba \Vert = 1}} \sum_{j=i}^{\ps} a_j^2 \slambdasub{i} {\bWsub_j}^\top \frac{\bAnsub}{k_{n}} \bWsub_j \\
    %& \hspace{5cm} + (1-c^2)   \Vert  \tbWnTsub  \frac{\bAnsub}{k_{n}} \tbWnsub \Vert \bigg\}\\
    &\leq \sup_{c \in [0,1]} \bigg\{ c^2  \sup_{ \substack{ \ba \in \R^{\ps - i+ 1} ,\\ \Vert \ba \Vert = 1} } \sum_{j=i}^{\ps} a_j^2   {\bWsub_j}^\top \frac{\dnsub \bAnsub}{k_{n}} \bWsub_j \\
    &\hspace*{5cm}+ (1-c^2) \frac{\dnsub}{\slambdasub{i}} \bigg\Vert  \tbWnTsub  \frac{\bAnsub}{k_{n}} \tbWnsub \bigg\Vert \bigg\}.
\end{align*}
 Note that \eqref{A.17} and $\tbWnsub$ being a orthogonal matrix imply that
\begin{equation*}
  \left\Vert  \tbWnTsub \left(\frac{\dnsub \bAnsub}{k_{n}} - \frac{\bBnsub}{k_{n}} \right) \tbWnsub \right\Vert \le   \left\Vert  \frac{\dnsub \bAnsub}{k_{n}} - \frac{\bBnsub}{k_{n}}   \right\Vert  \Pconv 0.
\end{equation*}
 We then conclude from $ \Vert  \tbWnTsub  \frac{\bBnsub}{k_{n}} \tbWnsub \Vert \Pconv (1 + \sqrt{c})^2$ (cf. proof of  \citet[Lemma 2.2 (i)]{Bai:Choi:Fujikoshi:2018}) and $\slambdasub{i} \rightarrow \infty$ that as $n\to\infty$,
\begin{equation*}
    \frac{1}{\slambdasub{i}}  \bigg\Vert  \tbWnTsub  \frac{\dnsub\bAnsub}{k_{n}} \tbWnsub \bigg\Vert \Pconv 0.
\end{equation*}
Additionally, with ${\bWsub_j}^\top \frac{\dnsub \bAnsub}{k_{n}} \bWsub_j \Pconv 1$ for $j = i+1, \ldots, \ps$ by 
\Cref{conv:hilf6} we get, %$\P$-a.s.
\begin{equation*}
   \frac{\dnsub \wellprime_{n,i}}{\slambdasub{i}}   \le   \frac{\dnsub }{\slambdasub{i}}   \sup_{\substack{\bw \perp \bWsub_1, \ldots, \bWsub_{i-1},\\ \Vert \bw \Vert = 1}} \! \bw^\top \bhSigmanprimesub \bw    
  \Pconv  \sup_{c \in [0,1]} \! \bigg\{ c^2  \sup_{ \substack{ \sum_{j=i}^{\ps} a_j^2 = 1}} \sum_{j=i}^{\ps} a_j^2  \bigg\} =1
\end{equation*}
as $n\to\infty$, which proves Step 2.
\end{proof}

\begin{lemma} \label{hilf:Theorem37}
       Let \Cref{asu:high_dim}  with  $\slambda{\ps}   \rightarrow \infty$ and $\slambda{1} = o(d_n^{1/2})$ as $\ninf$ be given. Then  as $n\to\infty$,
\begin{align*}
       \sup_{x \in ((1 - \sqrt{c})^2,(1 + \sqrt{c})^2)}\left| F^{\dnsub \bhSigmanprimesub}(x)- F_{c}(x)\right|  \Pconv 0, %\vert > \varepsilon \right) = 0, %\label{eq:mp_esd_conv}
\end{align*}
where $F^{d_n \bhSigmanprime}$ is  the empirical spectral distribution function of $d_n \bhSigmanprime$ and  $F_{c}(x)$ is defined as in \Cref{lem:Bai_Yaio2}.
\end{lemma}

\begin{proof}
For the ease of notation define the interval $I:=((1 - \sqrt{c})^2,(1 + \sqrt{c})^2)$.
Let $F^{  \tbWnTsub  \bBnsub \tbWnsub/k_{n}}$ and $F^{  \tbWnTsub  \dnsub\bAnsub \tbWnsub/k_{n}}$ be the empirical spectral  distribution function of $\tbWnTsub  \bBnsub \tbWnsub/k_{n}$ and $  \tbWnTsub  \dnsub\bAnsub \tbWnsub/k_{n}$, respectively. Due to  \eqref{A.17}, 
it follows by \citet[Theorem A.45]{Bai:Silverstein:2010}  that as $n\to\infty $, %for  any continuity point $x \in \R$ of $F_c$  
\begin{align*}
    \sup_{x \in I}\left|F^{  \tbWnTsub  \frac{\bBnsub}{k_{n}} \tbWnsub}(x) - F^{\tbWnTsub  \frac{\dnsub\bAnsub}{k_{n}} \tbWnsub}(x)\right|  \Pconv 0.
\end{align*} {By \citet[Theorem 1.1]{silverstein:1995}} and \citet[Theorem A.44]{Bai:Silverstein:2010} combined with $\rank(\bI - \bCnsub) = \rank( \frac{1}{{k_{n}}} \mathbf{1}_{k_{n}} \mathbf{1}_{k_{n}}^\top) = 1$ there exists a set $\Omega_0\in\mathcal{F}$ with $\P(\Omega_0)=1$ so that for any $\omega \in\Omega_0$ %and any continuity point $x \in \R$ of $F_c$ 
the convergence
\begin{align*}
    \lim_{n\to\infty}\sup_{x \in I}\left|F^{  \tbWnTsub  \frac{\bBnsub}{k_{n}} \tbWnsub}(x,\omega)- F_{c}(x)\right|=0 
\end{align*} 
holds which ends in 
\begin{align}
    \sup_{x \in I} \left| F^{ \tbWnTsub  \frac{\dnsub\bAnsub}{k_{n}} \tbWnsub}(x) - F_{c}(x) \right|\Pconv 0. 
\label{conv:mp_law_submatrix}
\end{align} 
Since the matrices $ {\bWsub}^\top \dnsub\bhSigmanprimesub \bWsub$ and $\dnsub \bhSigmanprimesub$ share the same eigenvalues \linebreak$\dnsub \wellprime_{n,\ps + 1}, \ldots, \dnsub \wellprime_{n, \dnsub}$,   we get for any $i \in \{\ps+1, \ldots,  \dnsub -\ps\}$ with the interlacing theorem for the eigenvalues (\citet[Theorem 4.3.28]{matrix_analysis}) $\P$-a.s. that
    \begin{align} %\wellprime_{n,i_n(\alpha)} \ge
       \lambda_{i} \bigg(   \tbWnTsub  \frac{ \dnsub\bAnsub}{k_{n}} \tbWnsub  \bigg)  &\ge \lambda_{\ps + i} \left(     {\bW}^\top  \dnsub\bhSigmanprimesub \bW   \right) \nonumber \\
       &=  \dnsub \wellprime_{n,\ps + i} \label{Aa}\\
       &\ge \lambda_{\ps + i} \bigg(   \tbWnTsub  \frac{\dnsub\bAnsub}{k_{n}} \tbWnsub \bigg). \nonumber
    \end{align}
Therefore, due to \Cref{conv:mp_law_submatrix} and \Cref{Aa}, %for all $\varepsilon > 0$
\begin{align*}
   \sup_{x \in I}\left| F^{\dnsub \bhSigmanprimesub}(x)- F_{c}(x)\right| &=\sup_{x \in I}\left|  \frac{1}{\dnsub  }\sum_{i=1}^{\dnsub} \mathbbm{1}  \left\{ \dnsub \wellprime_{n, i}\le x  \right\} - F_{c}(x)\right| \\
   & \leq   \sup_{x \in I} \left| F^{ \tbWnTsub  \frac{\dnsub\bAnsub}{k_{n}} \tbWnsub}(x) - F_{c}(x) \right|
   +\frac{4p^*}{d_n}\Pconv 0,
\end{align*}
which is the statement.
\end{proof}

\begin{proof}[Proof of \Cref{th:Bai_Lemma2_2_lambda}]
The proof of \Cref{th:Bai_Lemma2_2_lambda} (a)-(d) follows with the same arguments as the proof of \Cref{lem:Bai_Yaio} using only \Cref{hilf:Theorem36} and \Cref{hilf:Theorem37} in combination with $\bhSigman \overset{\mathcal{D}}{=} \bhSigmanprime$ (cf. \Cref{eq:cov_distr_equality}).
Only the proof (e) remains. Therefore, note that for $i < \ps$ the asymptotic behavior  $\frac{d_n  \well_{n,i}}{\slambda{i}} \Pconv 1$ and  $\frac{1}{d_n-i} \sum_{j=\ps + 1}^{d_n} \frac{d_n  \well_{n,j}}{\slambda{i}} \Pconv 0$ as $n \rightarrow \infty$ hold by (a) and (d), respectively. Hence,
\begin{align*}
     \frac{d_n  \well_{n,i}}{\frac{1}{d_n-i} \sum_{j=i+1}^{d_n} d_n \well_{n,j}} &= \frac{d_n  \well_{n,i}}{\frac{1}{d_n-i} \sum_{j=i+1}^{\ps} d_n \well_{n,j}   + \frac{1}{d_n-i} \sum_{j=\ps + 1}^{d_n} d_n \well_{n,j}} \\
     &\ge \frac{ \frac{d_n  \well_{n,i}}{\slambda{i}}   }{\frac{\ps - i }{d_n-i}  \frac{d_n  \well_{n,i}}{\slambda{i}}   + \frac{1}{d_n-i} \sum_{j=\ps + 1}^{d_n} \frac{d_n  \well_{n,j}}{\slambda{i}} } \Pconv \infty,
\end{align*}
which shows (e).
\end{proof}

\section{Proofs of Section 3}  \label{sec:fixed_dim_proof}

\begin{proof}[Proof of \Cref{th:AIC_Cons}]
Since by \Cref{rem:scale_invariant} (b) the $\AIC$ is scale invariant and hence, we assume w.l.o.g. that $\lambda = 1$. %, i.e. dividing $\bSigma$ by $\lambda$.
$\mbox{}$\\
\textbf{Step 1:} Suppose $p > {\ps}$. Note
\begin{equation*}
    2(p+1) (d - p/2 ) - (d-1)(d+2) = -  (d-p-2)(d-p+1).
\end{equation*}
By the definition of the $\AIC$ we obtain
 \begin{align}
         \AIC_{k_n}(p) -  \AIC_{k_n}(\ps) %&=  k_n \sum_{i=\ps+1}^{p} \log  ( \well_{n,i}  )  + k_n  (d - 1 - p)  \log \left( \frac{1}{d - 1-p} \sum_{j = p+1}^{d - 1} \well_{n,j} \right)  \nonumber \\
%    & \quad -  k_n (d - 1 - \ps)  \log \left( \frac{1}{d - 1-\ps} \sum_{j = \ps+1}^{d - 1} \well_{n,j} \right)    + 2(p+1) (d - p/2 ) - 2(\ps+1) (d - \ps/2 )\nonumber\\
         &=    k_n \sum_{i=\ps+1}^{p} \log  ( \well_{n,i}  )  + k_n  (d - 1 - p)  \log \left( \frac{1}{d - 1-p} \sum_{j = p+1}^{d - 1} \well_{n,j} \right)  \nonumber \\
    & \quad -  k_n (d - 1 - \ps)  \log \left( \frac{1}{d - 1-\ps} \sum_{j = \ps+1}^{d - 1} \well_{n,j} \right)  \nonumber \\
    &\quad-  (d-p-2)(d-p+1) + (d-\ps-2)(d-\ps+1), \nonumber%  \label{eq:AIC__cons1}
    \end{align}
where we used that $p > \ps$.   Inserting the alternative representation 
  \begin{equation*}
      (\well_{n,\ps +1}, \ldots, \well_{n,d} )^\top = \mathbf{1}_{d-\ps} + \frac{1}{\sqrt{k_n}} \bM_n,
  \end{equation*} 
where
\begin{equation*}
    \bM_n \coloneqq \sqrt{k_n} ( (\well_{n,\ps +1}, \ldots, \well_{n,d} )^\top -  \mathbf{1}_{d-\ps}),
\end{equation*}
gives that
\begin{align*}
     \AIC_{k_n}(p) -  \AIC_{k_n}(\ps)   &=    k_n \sum_{i=\ps+1}^{p} \log  \left( 1 + \frac{1}{\sqrt{k_n}} M_{n,i}  \right) \nonumber  \\
    & \qquad  +  k_n  (d - 1 - p)  \log \left(1 + \frac{1}{d - 1-p} \sum_{j = p+1}^{d - 1}  \frac{1}{\sqrt{k_n}} M_{n,j} \right)  \nonumber \\
    &  \qquad -  k_n (d - 1 - \ps)  \log \left(1 + \frac{1}{d - 1-\ps} \sum_{j = \ps+1}^{d - 1}  \frac{1}{\sqrt{k_n}} M_{n,j} \right)   \nonumber  \\
    & \qquad -   (d-p-2)(d-p+1) +  (d-\ps-2)(d-\ps+1).   \nonumber 
\end{align*}
Furthermore, $\bM_n = O_\P(1)$ due to \Cref{cor:ev_sqrtn_normality} (b). Additionally the Taylor expansion of the logarithm as $x \rightarrow 0$,
\begin{equation*}
\log(1 + x) = x - \frac{1}{2} x^2 + O(x^3), %\label{eq:log_asympt_expAIC}
\end{equation*}
gives that
 \begin{align*}
         \AIC&_{k_n}(p) -  \AIC_{k_n}(\ps)  \\
    &=    k_n \sum_{i=\ps+1}^{p}   \left(  \frac{1}{\sqrt{k_n}} M_{n,i} - \frac{1}{2} \frac{1}{ k_n } M_{n,i}^2 + O_\P(k_n^{-3/2}) \right)   \nonumber \\
         & \quad + k_n   \left(\sum_{j = p+1}^{d - 1}  \frac{1}{\sqrt{k_n}} M_{n,j}  - \frac{1}{2(d - 1-p)} \left(\sum_{j = p+1}^{d - 1}  \frac{1}{\sqrt{k_n}} M_{n,j} \right)^2 + O_\P(k_n^{-3/2}) \right)  \nonumber \\
    &  \quad -  k_n \left(\sum_{j = \ps+1}^{d - 1}  \frac{1}{\sqrt{k_n}} M_{n,j}  - \frac{1}{2(d - 1-\ps)} \left(\sum_{j = \ps+1}^{d - 1}  \frac{1}{\sqrt{k_n}} M_{n,j} \right)^2 + O_\P(k_n^{-3/2}) \right)  \nonumber  \\
    & \quad -  (d-p-2)(d-p+1) +  (d-\ps-2)(d-\ps+1)   \nonumber \\
   % &=    k_n \sum_{i=\ps+1}^{p}   \left( - \frac{1}{2} \frac{1}{ k_n } M_{n,i}^2 + O_\P(k_n^{-3/2}) \right)   \nonumber \\
    %     & \quad +  k_n   \left(  - \frac{1}{2(d - 1-p)} \left(\sum_{j = p+1}^{d - 1}  \frac{1}{\sqrt{k_n}} M_{n,j} \right)^2 + O_\P(k_n^{-3/2}) \right)  \nonumber \\
    %&  \quad -  k_n \left( - \frac{1}{2(d - 1-\ps)} \left(\sum_{j = \ps+1}^{d - 1}  \frac{1}{\sqrt{k_n}} M_{n,j} \right)^2 + O_\P(k_n^{-3/2}) \right)  \nonumber  \\
    %& \quad -  (d-p-2)(d-p+1) + \frac{1}{2} (d-\ps-2)(d-\ps+1)   \nonumber \\
    &=    - \frac{1}{2} \sum_{i=\ps+1}^{p}   M_{n,i}^2 - \frac{1}{2(d - 1-p)} \left(\sum_{j = p+1}^{d - 1}  M_{n,j} \right)^2   + \frac{1}{2(d - 1-\ps)} \left(\sum_{j = \ps+1}^{d - 1}  M_{n,j} \right)^2  \nonumber  \\
    & \quad -  (d-p-2)(d-p+1) +   (d-\ps-2)(d-\ps+1)  + O_\P(k_n^{-1/2}). \nonumber
\end{align*}
An application of \Cref{cor:ev_sqrtn_normality} (b) gives then
\begin{align} \label{eq:AIC_cons2}
\AIC&_{k_n}(p) -  \AIC_{k_n}(\ps) \\
&\Dconv - \frac{1}{2} \sum_{i=\ps+1}^{p}    M_{i}^2 - \frac{1}{2(d - 1-p)} \left(\sum_{j = p+1}^{d - 1}  M_{j} \right)^2    + \frac{1}{2(d - 1-\ps)} \left(\sum_{j = \ps+1}^{d - 1}  M_{j} \right)^2  \nonumber  \\
    & \qquad -  (d-p-2)(d-p+1) +   (d-\ps-2)(d-\ps+1). 
    \end{align}
Hence, the assertion follows.\\
\textbf{Step 2:} Suppose $p < {\ps}$. Again by the definition of the $\AIC$ we receive
    \begin{align}
     \frac{\AIC_{k_n}(p) - \AIC_{k_n}(\ps)}{k_n}
    &\;= - \sum_{j= p+1}^{\ps}  \log( \well_{n,j} ) + (d - 1 - p)  \log \left( \frac{1}{d - 1-p} \sum_{j = p+1}^{d - 1} \well_{n,j} \right) \nonumber  \\
    &\; \quad -(d - 1 - \ps)  \log \left( \frac{1}{d - 1-\ps} \sum_{j = \ps+1}^{d - 1} \well_{n,j} \right) \nonumber  \\
    & \; \quad - \frac{ (d-p-2)(d-p+1) + (d-\ps-2)(d-\ps+1)}{k_n}.\nonumber 
\end{align}
Due to \Cref{cor:ev_sqrtn_normality} (a),  $\well_{n,i} \Pconv \lambda_i$ for $i = 1, \ldots, d-1$ holds and therefore,
    \begin{eqnarray*}
     \frac{\AIC_{k_n}(p) - \AIC_{k_n}(\ps)}{k_n} &\Pconv &- \sum_{j=p+1}^{\ps}\log \left( \lambda_j  \right) + (d - 1 - p)  \log \big( \frac{1}{d - 1-p} \sum_{j = p+1}^{d - 1} \lambda_{j} \big) \nonumber \\
    %\end{align}
%Using $\lambda_i = 1$ for $i = \ps+1, \ldots, d-1$ we get
%\begin{align*}
%    - \sum_{j=p+1}^{\ps} &\log \left( \lambda_j  \right) + (d - 1 - p)  \log \big( \frac{1}{d - 1-p} \sum_{j = p+1}^{d - 1} \lambda_{j} \big) \\
    &=  &  - \sum_{j=p+1}^{d} \log \left( \lambda_j  \right) + (d - 1 - p)  \log \big( \frac{1}{d - 1-p} \sum_{j = p+1}^{d - 1} \lambda_{j} \big) \nonumber  \\
     &=  &  - \log \left(  \frac{\prod_{j=p+1}^{d} \lambda_j}{ \big( \frac{1}{d - 1-p} \sum_{j = p+1}^{d - 1} \lambda_{j} \big)^{(d - 1 - p)}}  \right) > 0, 
\end{eqnarray*}
due to the inequality of arithmetic and geometric means (\citet{AMGM}) which says that
\begin{equation*}
    \frac{ \big(\prod_{j=p+1}^{d} \lambda_j\big)^{1/(d - 1 - p)}}{  \frac{1}{d - 1-p} \sum_{j = p+1}^{d - 1} \lambda_{j}} < 1. \qedhere
\end{equation*}
\end{proof}

%\subsection{Proof of \Cref{th:BIC_Cons}}

\begin{proof}[Proof of \Cref{th:BIC_Cons}]
$\mbox{}$\\
\textbf{Step 1:} Suppose $p > {\ps}$. Due to \Cref{eq:AIC_cons2} we receive
\begin{align*}
    \BIC&_{k_n}(p) -  \BIC_{k_n}(\ps) \\
    &=\frac{\log(k_n)}{2} (d-\ps-2)(d-\ps+1)  - \frac{\log(k_n)}{2} (d-p-2)(d-p+1) + O_\P(1).
\end{align*}
A division by $\log(k_n)$ provides
\begin{equation*}
    \frac{\BIC_{k_n}(p) -  \BIC_{k_n}(\ps)}{\log(k_n)} \Pconv \frac{1}{2} (d-\ps-2)(d-\ps+1)  - \frac{1}{2} (d-p-2)(d-p+1),
\end{equation*}
which is strictly positive.\\
\textbf{Step 2:} Suppose $p < {\ps}$. Since 
as $\ninf$,
 \begin{align*}
        &\frac{\BIC_{k_n}(p) - \BIC_{k_n}(\ps)}{k_n} - \frac{\AIC_{k_n}(p) - \AIC_{k_n}(\ps)}{k_n} \\
        & \quad= \frac{\log(k_n) -2}{k_n} \left( \frac{(d-\ps-2)(d-\ps+1)}{2} -  \frac{(d-p-2)(d-p+1)}{2} \right) \rightarrow 0,
    \end{align*}
the statement follows from \Cref{th:AIC_Cons}.

\end{proof}

\section{Proofs of Section 4} \label{sec:proofs:C}

\begin{proof}[Proof of \Cref{th:cons_aic_high_dim}]
Note, as stated in \Cref{rem:scale_invariant}, the information criteria are scale invariant and hence
\begin{equation*} %\label{eq:aic_scale_invariant}
    \QAIC_{k_n}(\pn; \well_{n,1}, \ldots, \well_{n,d_n-1})=:\QAIC_{k_n}(\pn)= \QAIC_{k_n}(\pn; d_n \well_{n,1}, \ldots, d_n \wellprime_{n,d_n-1}).
\end{equation*}
Due to \Cref{lem:Bai_Yaio2} for (a,b) and \Cref{th:Bai_Lemma2_2_lambda} for (c), the proof of \citet[Theorem 3.1]{Bai:Choi:Fujikoshi:2018}
for $\hslambda_{n,1}, \ldots, \hslambda_{n,d_n-1}$ can be carried out step by step for $d_n \well_{n,1}, \ldots, d_n \well_{n,d_n-1}$.
 The only difference is that there we have almost sure convergence and here we have convergence in probability.
\end{proof}

\begin{proof}[Proof of \Cref{th:cons_bic_high_dim}]
Due to the scale invariance of the  $\QBIC_{k_n}(\ps)$, $\log(d_n)/\log(k_n)$ $\to 1$ as $\ninf$, \Cref{lem:Bai_Yaio2}  and \Cref{th:Bai_Lemma2_2_lambda}, the proof of \citet[Theorem 3.2]{Bai:Choi:Fujikoshi:2018}
for $\hslambda_{n,1}, \ldots, \hslambda_{n,d_n-1}$ 
can be carried out step by step for $d_n \well_{n,1}, \ldots, d_n \well_{n,d_n-1} $. 
\end{proof}

\begin{proof}[Proof of \Cref{th:cons_high_dim_greater_1}]
Due to the scale invariance of the  $\QAICstar$, \Cref{lem:Bai_Yaio2}  and \Cref{th:Bai_Lemma2_2_lambda}, the proof of \citet[Theorem 3.3]{Bai:Choi:Fujikoshi:2018}
for $\hslambda_{n,1}, \ldots, \hslambda_{n,d_n-1}$ can be carried out step by step for $d_n \well_{n,1}, \ldots, d_n \well_{n,d_n-1} $. 
\end{proof}

\begin{proof}[Proof of \Cref{th:cons_bic_high_dim_greater_1}]
Due to the scale invariance of the  $\QBICstar$, $\log(d_n)/\log(k_n)$ $\to 1$ as $\ninf$, \Cref{lem:Bai_Yaio2} and \Cref{th:Bai_Lemma2_2_lambda}, the proof of \citet[Theorem 3.4]{Bai:Choi:Fujikoshi:2018}
for $\hslambda_{n,1}, \ldots, \hslambda_{n,d_n-1}$ 
can be carried out step by step for $d_n \well_{n,1}, \ldots, d_n \well_{n,d_n-1} $. 
\end{proof}

\section{Proofs of Section 5}

\begin{lemma} \label{lemma:norm}
Let
\begin{equation*}
    \bepsilon_d  \sim \Big\vert \mathcal{N}_{d} \Big( \mathbf{0}_d, \frac{100}{d} \bI_d \Big)\Big\vert,
\end{equation*}
where the absolute value is entry-wise. Then
\begin{equation*}
    \lim_{d\to\infty}\V (\Vert \bepsilon_d \Vert)=100/\sqrt{2}.
\end{equation*}
\end{lemma}

\begin{proof}

Indeed, since $\Vert \bepsilon_d \Vert \sim \sqrt{100/d} \sqrt{\chi^2_d}$, where $\chi^2_d$ is a chi-square distribution with $d$ degrees of freedom, the formula for the moments of a chi-square distribution (cf. Theorem 3.3.2 in \citet{Hogg:2005}) gives
\begin{equation*}
    \V (\Vert \bepsilon_d \Vert) =  \E[\Vert \bepsilon_d \Vert^2] - ( \E[ \Vert \bepsilon_d \Vert ]) ^2   = \frac{100}{d} \Big( d - \Big(  \frac{\sqrt{2} \Gamma( (d+1)/2)}{ \Gamma(d/2)} \Big) ^2   \Big).
\end{equation*}
Further by Gautschi's inequality (cf. \citet[p. 1]{Gautschi}) we have
\begin{equation*}
    \left( \frac{d-1}{2} \right)^{1/2} \le   \frac{ \Gamma( (d+1)/2)}{ \Gamma(d/2)} \le \left( \frac{d-1}{2} + 1 \right)^{1/2}
\end{equation*}
and therefore
\begin{align*}
   \sqrt{2} \cdot 100  \frac{d-1}{2d} &=  \frac{100}{d} \sqrt{2}\Big( d - \frac{d-1}{2} - 1   \Big)  \\
   &\le  \V (\Vert \bepsilon_d \Vert) \le  \frac{100}{d} \sqrt{2} \Big( d -  \frac{d-1}{2}   \Big) = \sqrt{2} \cdot 100  \frac{d+1}{2d}.
\end{align*}
Letting $d\to \infty$ on the left and on the right-hand side gives the statement.
\end{proof}

% \putbib
% % \putbib
  \end{bibunit}
\end{document}